\pdfoutput=1
\newif\iffull \fulltrue

\documentclass{sig-alternate-05-2015}

\usepackage[numbers,sort&compress]{natbib}
\usepackage[english]{babel}
\usepackage[T1]{fontenc}
\usepackage[utf8]{inputenc} 
\usepackage{upgreek}
\usepackage{paralist}
\usepackage{enumitem}
\usepackage{url}
\usepackage{stmaryrd}
\usepackage{graphicx}
\usepackage{times}

\usepackage{xspace}

\usepackage{amsmath,amsfonts,amsthm,amssymb,mathtools,thmtools}
\usepackage{xcolor}
\usepackage[ruled,vlined]{algorithm2e}
\usepackage{nicefrac}
\usepackage{bbm}

\newcommand{\EasyCrypt}{\textsf{EasyCrypt}\xspace}

\newcommand{\Sahl}{\textsf{aHL}\xspace}

\newcommand{\Sprhl}{\textsf{pRHL}\xspace}
\newcommand{\Saprhl}{\textsf{apRHL}\xspace}
\newcommand{\Saprhlp}{\textsf{apRHL}$^+$\xspace}

\newcommand{\EXAMPLE}{$\mathsf{ASV}_{\mathrm{bt}}$\xspace}
\usepackage{mathpartir}
\usepackage{yfonts}

\usepackage{todonotes}

\usepackage{xargs}
\newcommandx{\note}[2][1=]{\todo[inline,linecolor=red,backgroundcolor=red!25,bordercolor=red,#1]{#2}}

\newcommand{\Dist}{\ensuremath{\mathbf{Distr}}}

\newcommand{\supp}{\mathsf{supp}}

\newcommand{\subst}[2]{\left\{#2/#1\right\}}

\newcommand{\Var}{\mathcal{X}}

\newcommand{\Mem}{\mathsf{State}}

\newcommand\q{[\![}
\newcommand\p{]\!]}

\newcommand{\ahl}[4]{\vdash_{#4} #1 : #2 \Longrightarrow #3}

\newcommand{\AEquiv}[6]{\vdash {#1} \sim_{\!\left\langle#5,#6\right\rangle} {#2} : {#3} \Longrightarrow {#4}}
\newcommand{\sidel}{\langle 1\rangle}
\newcommand{\sider}{\langle 2\rangle}

\newcommand\denot[1]{\q #1 \p}
\newcommand{\dsem}[2]{\denot{#2}_{#1}}

\newcommand{\Skip}{\mathsf{skip}}
\newcommand{\Seq}[2]{{#1};\,{#2}}
\newcommand{\Ass}[2]{#1 \leftarrow #2}
\newcommand{\Rand}[2]{#1 \stackrel{\raisebox{-.25ex}[.25ex]%
 {\tiny $\mathdollar$}}{\raisebox{-.2ex}[.2ex]{$\leftarrow$}} #2}
\newcommand{\Cond}[3]{\mathsf{if}\ #1\ \mathsf{then}\ #2\ \mathsf{else}\ #3}
\newcommand{\Condt}[2]{\mathsf{if}\ #1\ \mathsf{then}\ #2}

\newcommand{\WWhile}[2]{\mathsf{while}\ #1\ \mathsf{do}\ #2}
\newcommand{\Call}[3]{#1 \leftarrow #2\mathsf{(}#3\mathsf{)}}

\newcommand{\Lap}{\mathcal{L}}

\newcommand{\Expr}{\mathcal{E}}

\newcommand{\Cmd}{\mathcal{C}}

\newcommand{\A}{\mathcal{A}}
\newcommand{\Or}{\mathcal{O}}

\newcommand{\lift}[1]{\mathrel{#1^\sharp}}
\newcommand{\alift}[2]{\mathrel{#1^{\sharp #2}}}
\usepackage{listings}

\usepackage[T1]{fontenc}
\usepackage{microtype}

\usepackage[scaled]{beramono}
\newcommand\Small{\fontsize{8.2pt}{8.4pt}\selectfont}
\newcommand*\LSTfont{\Small\ttfamily\SetTracking{encoding=*}{-60}\lsstyle}

\def\lstrnd{\stackrel{\raisebox{-.15ex}{\ensuremath{\scriptscriptstyle\$}}}{\raisebox{-.2ex}{\ensuremath{\leftarrow}}}}

\usepackage{xcolor}
\definecolor{DarkGreen}{rgb}{0.1,0.5,0.1}
\definecolor{DarkRed}{rgb}{0.5,0.1,0.1}
\definecolor{DarkBlue}{rgb}{0.1,0.1,0.5}
\usepackage{hyperref}
\hypersetup{
    unicode=false,          
    pdftoolbar=true,        
    pdfmenubar=true,        
    pdffitwindow=false,      
    pdftitle={},    
    pdfauthor={}
    pdfsubject={},   
    pdfnewwindow=true,      
    pdfkeywords={keywords}, 
    colorlinks=true,       
    linkcolor=DarkRed,          
    citecolor=DarkGreen,        
    filecolor=DarkRed,      
    urlcolor=DarkBlue,          
}
\usepackage[capitalise]{cleveref}

\declaretheorem{theorem}
\declaretheorem[numberlike=theorem]{definition}
\declaretheorem[numberlike=theorem]{lemma}

\declaretheorem[numberlike=theorem]{proposition}

\crefname{section}{\S}{\S}
\Crefname{section}{\S}{\S}

\crefname{prop}{proposition}{propositions}
\Crefname{prop}{Proposition}{Propositions}

\crefname{lem}{lemma}{lemmas}
\Crefname{lem}{Lemma}{Lemmas}

\crefname{thm}{theorem}{theorems}
\Crefname{thm}{Theorem}{Theorems}

\crefname{definition}{definition}{definitions}
\Crefname{definition}{Definition}{Definitions}

\newcommand{\thefullversion}{\iffull the appendix\else the full version\fi}

\title{Advanced Probabilistic Couplings for Differential Privacy\iffull\else\titlenote{%
  The full version of this paper is available at
  \mbox{\url{https://arxiv.org/abs/1606.07143}}.}\fi}

\subtitle{Accuracy-dependent, Advanced Composition, Adaptive Algorithms}

\numberofauthors{6}
\author{
  \alignauthor
  Gilles Barthe \\
    \affaddr{IMDEA Software Institute} \\
    \affaddr{Madrid, Spain}
  \alignauthor
  Noémie Fong \\
    \affaddr{ENS} \\
    \affaddr{Paris, France}
  \alignauthor
  Marco Gaboardi\titlenote{Partially supported by NSF grants CNS-1237235,
    CNS-1565365 and by EPSRC grant EP/M022358/1.} \\
    \affaddr{University at Buffalo, SUNY} \\
    \affaddr{Buffalo, USA}
  \and
  \alignauthor
  Benjamin Grégoire \\
    \affaddr{Inria} \\
    \affaddr{Sophia-Antipolis, France}
  \alignauthor
  Justin Hsu\titlenote{%
    Partially supported by NSF grants $\#1065060$ and $\#1513694$, and a grant
    from the Simons Foundation ($\#360368$ to Justin Hsu).}\\
    \affaddr{University of Pennsylvania} \\
    \affaddr{Philadelphia, USA}
  \alignauthor
  Pierre-Yves Strub \\
    \affaddr{IMDEA Software Institute} \\
    \affaddr{Madrid, Spain}
}

\begin{document}

\CopyrightYear{2016} 
\setcopyright{acmlicensed}
\conferenceinfo{CCS'16,}{October 24 - 28, 2016, Vienna, Austria}
\isbn{978-1-4503-4139-4/16/10}\acmPrice{\$15.00}
\doi{http://dx.doi.org/10.1145/2976749.2978391}

\maketitle

\begin{abstract}
Differential privacy is a promising formal approach to data privacy, which
provides a quantitative bound on the privacy cost of an algorithm that operates
on sensitive information.  Several tools have been developed for the formal
verification of differentially private algorithms, including program logics and
type systems.  However, these tools do not capture fundamental techniques that
have emerged in recent years, and cannot be used for reasoning about
cutting-edge differentially private algorithms. Existing techniques fail to
handle three broad classes of algorithms: 1) algorithms where privacy depends on
accuracy guarantees, 2) algorithms that are analyzed with the advanced
composition theorem, which shows slower growth in the privacy cost, 3)
algorithms that interactively accept adaptive inputs.

We address these limitations with a new formalism extending
\Saprhl~\citep{BartheKOZ13}, a relational program logic that has been
used for proving differential privacy of non-interactive algorithms,
and incorporating \Sahl~\citep{BGGHS16-icalp}, a (non-relational)
program logic for accuracy properties.  We illustrate our approach
through a single running example, which exemplifies the three classes
of algorithms and explores new variants of the Sparse Vector
technique, a well-studied algorithm from the privacy literature.  We
implement our logic in \EasyCrypt, and formally verify privacy. We
also introduce a novel coupling technique called \emph{optimal subset
  coupling} that may be of independent interest.
\end{abstract}

\section{Introduction}

\emph{Differential privacy}, a rigorous and quantitative notion of
statistical privacy, is one of the most promising formal definitions
of privacy to date.  Since its initial formulation by \citet{DMNS06},
differential privacy has attracted substantial attention throughout
computer science, including areas like databases, machine learning,
and optimization, and more.

There are several reasons for this success. For one, differential privacy allows
a formal trade-off between privacy and accuracy: differentially private
algorithms come with a \emph{privacy guarantee} expressed in terms of two
parameters $\epsilon$ (expressing the privacy cost) and $\delta$ (expressing the
probability of violating the privacy cost). For both parameters, smaller values
offer stronger privacy guarantees.

Another important advantage differential privacy is that it \emph{composes}
well: differentially private algorithms can be easily combined to build new
private algorithms. The differential privacy literature offers several
\emph{composition theorems}, differing in how the privacy parameter of the
larger algorithm depends on the parameters of the components. These composition
properties can also be used in \emph{interactive and adaptive} scenarios where
an adversary can decide which algorithm to run depending on the outputs of
previous algorithms.

Differential privacy's clean composition properties also make it an attractive
target for an unusually diverse array of formal verification techniques.  By
now, there are tools that formally guarantee differential privacy via
relational program logics~\citep{BartheKOZ13,BartheO13}, linear type
systems~\citep{ReedPierce10,GHHNP13,EignerM13}, interactive
automata~\citep{Tschantz201161,xu:hal-00879140}, product
programs~\citep{BGGHKS14}, satisfiability modulo theories~\citep{FredriksonJ14},
refinement type systems~\citep{BGGHRS15}, and more.  While these systems
formalize privacy through a wide variety techniques, most of these approaches
analyze a composition of private algorithms using the \emph{sequential
  composition} theorem of differential privacy, which guarantees that the
resulting algorithms have parameter $\epsilon$ and $\delta$ equal to the
\emph{sum} of the parameters of their components.

Recently, \citet{BGGHS16} highlighted a close connection between
\emph{approximate couplings} and differential privacy which enables formal
verification beyond sequential composition.  \citet{BGGHS16} work with the
relational program logic \Saprhl~\citep{BartheKOZ13}, extended with a new
composition principle called \emph{pointwise privacy}. The idea is to first
prove a restricted case of privacy---corresponding roughly to privacy for a
single output---and then combine the results to prove the full differential
privacy property.  Combined with the composition principle for approximate
couplings, which generalizes the sequential composition theorem, \Saprhl can
express simple, compositional proofs of algorithms like the one-side Laplace
implementation of the Exponential mechanism~\citep{DR14} and the Above Threshold
algorithm~\citep{DR14} while abstracting away reasoning about probabilistic
events. Existing privacy proofs for these algorithms, even on paper, involve ad
hoc computation of probabilities.

While \Saprhl substantially expands the range of formal verification in
differential privacy, there are still private algorithms that \Saprhl cannot
verify. Roughly, there are three missing pieces:
\begin{itemize}
  \item \emph{Accuracy-dependent privacy}. Some algorithms are only private if
    an accuracy property holds.
  \item \emph{Advanced composition}.  This principle shows slower growth in the
    privacy cost, in exchange for a small probability of violating privacy. The
    proof involves an technical martingale concentration argument.
  \item \emph{Interactive privacy}. Some private algorithms are
    \emph{interactive}, receiving a continuous sequence of \emph{adaptive}
    inputs while producing intermediate outputs.
\end{itemize}

These three missing pieces correspond to three fundamental principles of
differential privacy. While there are many algorithms from the privacy
literature that use one (or more) of these three features, to structure our
presentation we will work with a variant of the Sparse Vector technique based on
the Between Thresholds algorithm \citep{BunSU16}, a single unifying example that
uses all three features (\cref{sec:motivating}). After reviewing some technical
preliminaries about differential privacy, approximate couplings, and the logic
\Saprhl (\cref{sec:prelim}), we describe extensions to \Saprhl to verify privacy
for new classes of algorithms.

\begin{itemize}
  \item New proof rules that allow reasoning within \Saprhl
    while assuming an accuracy property, incorporating accuracy proofs from the
    Hoare logic \Sahl~\citep{BGGHS16-icalp} (\cref{sec:accpriv}). We demonstrate
    these rules on a classic example of accuracy-dependent privacy: the
    Propose-test-release framework~\citep{DworkL09,ThakurtaS13}.
  \item A proof rule that analyzes loops using the advanced composition
    principle; soundness relies on a novel generalization of advanced
    composition to approximate couplings (\cref{sec:advcomp}).
  \item New proof rules for \emph{adversaries}, external procedure calls
    that model an adaptive source of inputs (\cref{sec:interactive}).
  \item Orthogonal to reasoning about accuracy, advanced composition, and
    adversarial inputs, we introduce a general construction that may be of
    independent interest called the \emph{optimal subset coupling}. This
    construction gives an approximate lifting relating subsets that yields the
    best possible $\epsilon$, and we use this construction to give a new
    interval coupling rule for the Laplace distribution (\cref{sec:optimal}).
\end{itemize}

We then show how to combine these ingredients to verify our main running
example, the Between Thresholds algorithm (\cref{sec:bt}).  We finish with
related work (\cref{sec:related}) and some concluding thoughts
(\cref{sec:conclusions}).

\section{Motivating example}\label{sec:motivating}

Before diving into the technical details we'll first present our motivating
example, which involves accuracy-dependent privacy, advanced composition, and
interactive privacy. We first review the definition of differential privacy, a
relational property about probabilistic programs proposed by Dwork, McSherry,
Nissim and Smith.

\begin{definition}
  Let the \emph{adjacency relation} be $\Phi \subseteq A \times A$, and
  $\epsilon, \delta > 0$.  A program $M:A\rightarrow\Dist(B)$ satisfies
  $(\epsilon,\delta)$-\emph{differential privacy} with respect to $\Phi$ if for
  every pair of inputs $a, a'\in A$ such that $\Phi(a, a')$ and every subset of
  outputs $S \subseteq B$, we have
  \[
    \Pr_{y\leftarrow M a}[y \in S]
    \leq \exp(\epsilon) \Pr_{y\leftarrow M a'} [y \in S] + \delta .
  \]
  When $\delta = 0$, we say that $M$ is $\epsilon$-\emph{differentially private}.
\end{definition}

Intuitively, $\Phi$ relates inputs that differ in a single individual's data.
Then, differential privacy requires that the two resulting distributions on
outputs should be close.

Our motivating example is \emph{Adaptive Sparse Vector for Between
  Thresholds} (\EXAMPLE), a variant of the Sparse Vector algorithm.
Our algorithm takes a stream of numeric queries as input, and answers
only the queries that take a value within some range. The main benefit
of Sparse Vector is that queries that take a value outside the range do
not increase the privacy cost, even though testing whether whether the
query is (approximately) in the range involves private data. Sparse
Vector is an appealing example, because of its popularity and its
difficulty. In particular, the privacy proof of Above Threshold is
non-compositional and notoriously tricky, and several
variants\footnote{There exist multiple versions of Sparse Vector. The
  earliest reference seems to be \citet{DNRRV09}; several refinements
  were proposed by \citet{RR10,HR10}. Applications often use their own
  variants, e.g.~\citet{ShokriS15}. The most canonical version of the
  algorithm is the version by~\citet{DR14}.} of the
algorithm were supposedly proved to be private but were later shown to
be non-private (\citet{lyu2016understanding} provide a comprehensive
survey).

The code of \EXAMPLE is shown in \cref{fig:running}. At a high level,
the algorithm accepts a stream of adversarially chosen queries and
produces a list of queries whose answer lies (approximately) between
two threshold parameters $a$ and $b$.  The algorithm computes noisy
versions $A$ and $B$ of $a$ and $b$ using the Laplace mechanism
$\Lap_{\epsilon}$, which we review in the next section, and then
performs an interactive loop for a fixed number ($N$) of
iterations. Each step, a stateful adversary $\mathcal{A}$ receives the
current list $l$ of queries whose answer on input database $d$ lies
between $[A,B]$ and selects a new query $q$. If its noisy answer $S$
lies in $[A, B]$ and there have been fewer than $M$ queries between
threshold, the algorithm adds $q$ to the list $l$. Our algorithm
differs from standard presentations of Adaptive Sparse
Vector~\citep{DR14} in two significant respects:
\begin{itemize}
\item we use \textsf{BetweenThresholds} rather than
  \textsf{AboveThreshold} for deciding whether to add a query to the
  list;
\item we do not rerandomize the noise on the thresholds each time a
  query is added to $l$; therefore, our algorithm adds less noise.
\end{itemize}
\begin{figure}[t]
\[
  \begin{array}{l} \\
  \mbox{\EXAMPLE}(a,b,M,N,d) := \\
    \Ass{i}{0};
    \Ass{l}{[]}; \\
    \Rand{u}{\Lap_{\epsilon/2}(0)}; \\
    \Ass{A}{a-u};
    \Ass{B}{b+u}; \\
    \WWhile{i<N \land |l|< M}{} \\
    \quad  \Call{q}{\mathcal{A}}{l};\\
    \quad  \Rand{S}{\Lap_{\epsilon'/3}(\mathsf{evalQ}(q,d))};  \\
    \quad  \mathsf{if}~ (A\leq S \leq B)~
     \mathsf{then}~\Ass{l}{i::l}; \\
    \quad  \Ass{i}{i + 1}; \\
    \mathsf{return}~l
  \end{array}
\]
\caption{Sparse Vector for Between Thresholds}
\label{fig:running}
\end{figure}

\EXAMPLE satisfies the following privacy guarantee.
\begin{theorem}\label{main:thm}
  Let $\epsilon$ and $\delta$ both be in $(0, 1)$. Set
  \[
    \epsilon' \triangleq \frac{\epsilon}{4 \sqrt{2 M \ln(2/\delta)}} ,
  \]
  and suppose $a$ and $b$ are such that
\[
    b - a \geq \frac{6}{\epsilon'} \ln (4/\epsilon') + \frac{4}{\epsilon}
    \ln(2/\delta) .
  \]
If all adversarial queries $q$ are $1$-sensitive (i.e.\,
$|\mathsf{evalQ}(q,d)-\mathsf{evalQ}(q,d')| \leq 1$ for every adjacent
databases $d$ and $d'$), then \EXAMPLE is $(\epsilon,\delta)$-differentially
private.
\end{theorem}
The formal proof of this theorem, which we have verified in an implementation of
our logic within the \EasyCrypt system, involves several features:
\begin{itemize}
 \item reasoning principles for mixing accuracy and privacy
   guarantees, using a combination of relational
   logics~\citep{BartheKOZ13,BGGHS16} and non-relational
   logics~\citep{BGGHS16-icalp};
\item a generalization of the advanced composition theorem for
  handling the body of the loop;
\item an adversary rule for handling interactive inputs in the loop; and
\item a new reasoning principle, called \emph{optimal subset coupling}, for
  handling the Laplace mechanism in the loop.
\end{itemize}
We stress that the use of pointwise equality, which is required for
proving privacy of between thresholds, makes the proof significantly
more challenging than other examples involving solely adaptive adversaries,
advanced composition, and accuracy-dependent privacy.

We remark that \citet{BunSU16} proposed Between Threshold and proved its
privacy. Their proof does not use advanced composition, and follows from a
somewhat complicated calculations about the probabilities of certain events. Our
proof demonstrates the power of approximate liftings: somewhat surprisingly, we
arrive at an elegant privacy proof without probabilistic reasoning.

\section{Background} \label{sec:prelim}
Before presenting our new extensions, we first review some
preliminaries about differential privacy, the connection to
approximate liftings, the program logic \Saprhl \citep{BartheKOZ13}
and its extension \Saprhlp \citep{BGGHS16}, and the union bound logic
\Sahl \citep{BGGHS16-icalp}.

\subsection{Mathematical preliminaries}

To avoid measure-theoretic issues, we base our technical development on
distributions over discrete sets $B$. A function $\mu:B\rightarrow
\mathbb{R}^{\geq 0}$ is a \emph{distribution} over $B$ if $\sum_{b\in\supp(\mu)}
\mu(b)= 1$.  As usual, the \emph{support} $\supp(\mu)$ is the subset of $B$ with
non-zero weight under $\mu$.  We write $\Dist(B)$ for the set of discrete
distributions over $B$.  Equality of distributions is defined as pointwise
equality of functions.

We will also use \emph{marginal distributions}. Formally, the first and second
marginals of a distribution $\mu\in\Dist(B_1\times B_2)$ are simply the
projections: the distributions $\pi_1(\mu)\in\Dist(B_1)$ and
$\pi_2(\mu)\in\Dist(B_2)$ given by
\[
\pi_1(\mu)(b_1)=\sum_{b_2\in B_2} \mu(b_1,b_2) \qquad
\pi_2(\mu)(b_2)=\sum_{b_1\in B_1} \mu(b_1,b_2) .
\]

\subsection{Differential privacy}

We will need several tools from differential privacy;
readers should consult the textbook by \citet{DR14} for a more comprehensive
introduction.  Most differentially private algorithms are constructed from
private primitive operations. The most famous primitive is the Laplace
mechanism.

\begin{definition}[Laplace mechanism \citep{DMNS06}]
  Let $\epsilon>0$. The \emph{(discrete) Laplace mechanism}
  $\Lap_\epsilon:\mathbb{Z}\rightarrow \Dist(\mathbb{Z})$ is defined by
  $\Lap_{\epsilon}(t) = t + \nu$, where $\nu \in \mathbb{Z}$ with probability
  proportional to
  \[
    \Pr[ \nu ] \propto \exp{(-\epsilon\cdot |\nu|)}.
  \]
\end{definition}
The level of privacy depends on the sensitivity of the query, which measures how
far the function may differ on two related inputs. Roughly, adding the same
level of Laplace noise to a higher sensitivity query will be less private.
\begin{definition}[Sensitivity]
  A function $F:A\rightarrow\mathbb{Z}$ is \emph{$k$-sensitive with respect to
    $\Phi\subseteq A\times A$} if $|F(a_1) - F(a_2)| \leq k$ for every
  $a_1,a_2\in A$ such that $\Phi(a_1, a_2)$.
\end{definition}

The Laplace mechanism satisfies an accuracy specification.

\begin{proposition}[Laplace accuracy] \label{prop:lap-acc}
  Let $\epsilon, \beta > 0$, and suppose $x$ is the result from running
  $\Lap_{\epsilon}(t)$. Then $|x - t| \leq \frac{1}{\epsilon} \ln
  \frac{1}{\beta}$ with probability at least $1 - \beta$.
\end{proposition}

Besides private primitives, the other main tools for constructing
private programs are the \emph{composition theorems}. These results
describe the privacy level for a combination of private
programs---say, calls to the Laplace mechanism.  We will use a bit of
notation for compositions. Let $\{ f_i \}$ be a set of $n$ functions
of type $A \to D \to \Dist(A)$. Denote the $n$-fold composition $f^n :
A \to D \to \Dist(A)$ by
\[
  f^k(a, d) =
  \begin{cases}
    \mathsf{unit}\ a &: k = 0 \\
    \mathsf{bind}\ f^{k - 1}(a, d)\ f_k (-, d) &: k \geq 1 .
  \end{cases}
\]
Here, $\mathsf{unit} : A \to \Dist(A)$ and $\mathsf{bind} : \Dist(A)
\to (A \to \Dist(B)) \to \Dist(B)$ are the monadic operations for
distributions. They satisfy the following equalities:
$$\mathsf{unit}(a)(b)=
\left\{\begin{array}{ll}
1 & \mbox{if}~a=b \\
0 & \mbox{otherwise}
\end{array}\right.
$$
and for $f:A \to\Dist(B)$ and $F:A\rightarrow B\rightarrow\Dist(C)$ ,
$$
(\mathsf{bind}~f~F)(a)(c)=\sum_{b} f(a)(b)~F(a)(b)(c) .
$$
We will also use this composition notation when the functions $\{ f_i \}$ have
type $A \to \Dist(A)$, simply dropping the parameter $d$ above when defining
$f^n : A \to \Dist(A)$.

Then, the most basic composition theorem in differential
privacy is \emph{sequential composition}.

\begin{theorem}[Sequential composition]
  \label{thm:seqcomp-dp}
Let $f_i : A \to D \to \Dist(A)$ be a sequence of $n$ functions, such that
for every fixed $a \in A$, the functions $f_i(a) : D \to \Dist(A)$ are
$(\epsilon_i, \delta_i)$-differentially private for some adjacency relation on $D$.
Then for every initial value $a \in A$, the composition $f^n(a) : D \to
\Dist(A)$ is $(\epsilon^*, \delta^*)$-differentially private for
\[
  \epsilon^* = \sum_{i = 1}^n \epsilon_i
  \quad \text{and} \quad
  \delta^* = \sum_{i = 1}^n \delta_i .
\]
That is, the epsilons and deltas sum up through composition.
\end{theorem}
The sequential composition theorem is quite useful, and is the main principle
supporting modular verification of differential privacy. However, there is
another composition theorem, known as \emph{advanced composition} \citep{DRV10}.
Instead of summing up the privacy costs, this theorem gives slower growth of
$\epsilon$ in exchange for increasing the $\delta$ parameter. Advanced
composition is an extremely common tool for analyzing differentially private
algorithms, but it is not supported by most formal verification systems today.

\begin{theorem}[Advanced composition]
  \label{thm:advcomp-dp}
  Let $f_i : A \to D \to \Dist(A)$ be a sequence of $n$ functions, such that
  for every fixed $a \in A$, the functions $f_i(a) : D \to \Dist(A)$ are
  $(\epsilon, \delta)$-differentially private for some adjacency relation on $D$.
  Then, for every $a \in A$ and $\omega \in (0, 1)$, the composed function
  $f^n(a) : D \to \Dist(A)$ is $(\epsilon^*, \delta^*)$-differentially private for
  \[
    \epsilon^* = \left(\sqrt{2 n \ln(1/\omega)}\right) \epsilon +
    n \epsilon(e^\epsilon - 1)
    \quad \text{and} \quad
    \delta^* = n \delta + \omega .
  \]
  In particular, if we have $\epsilon' \in (0, 1)$, $\omega \in (0, 1/2)$, and
  \[
    \epsilon = \frac{\epsilon'}{2 \sqrt{2 n \ln(1/\omega)}} ,
  \]
  a short calculation\iffull\footnote{%
    First, note that $e^\epsilon - 1 \leq 2 \epsilon$ for $\epsilon \in (0, 1)$;
    this follows by convexity of $e^\epsilon - 2\epsilon - 1$.  Then, we have:
    \begin{align*}
      \sqrt{2 n \ln(1/\omega)} \epsilon + n \epsilon(e^\epsilon - 1)
      &\leq \sqrt{2 n \ln(1/\omega)} \epsilon + 2 n \epsilon^2 \\
      &= \frac{\epsilon'}{2} + \frac{\epsilon'}{2} \cdot \frac{\epsilon'}{2 \ln(1/\omega)} \\
      &\leq \frac{\epsilon'}{2} + \frac{\epsilon'}{2} = \epsilon',
    \end{align*}
    where the last inequality is because $\omega \in (0, 1/2)$ and $\epsilon'
    \in (0, 1)$, and the last factor is maximized at $\epsilon' = 1$ and $\omega
    = 1/2$:
    \[
      \frac{\epsilon'}{2 \ln(1/\omega)} \leq \frac{1}{2 \ln (2)} < 1 .
    \]}\fi{}
  shows that the function $f^n$ is $(\epsilon', \delta^*)$-differentially
  private.
\end{theorem}
\citet{DBLP:journals/corr/OhV13} propose sharper versions of this composition
theorem, including a provably optimal version and a version for the heterogeneous
case when the privacy level $(\epsilon_i, \delta_i)$ may depend on $i$. We
will use \cref{thm:advcomp-dp} for simplicity, but our techniques enable other
composition theorems to be easily plugged in.

\subsection{Approximate liftings}
While the definition of differential privacy seems to be a straightforward
property about probabilities in two distributions, a recent line of work
initiated by \citet{BartheKOZ13} and subsequently developed
\citep{BartheO13,BGGHS16} shows that differential privacy is a consequence of an
approximate version of probabilistic coupling, called \emph{approximate
  liftings}. Couplings are a long-standing tool in probability theory for
analyzing pairs of distributions, but the relation between differential privacy
and approximate couplings is still being explored.

Unlike couplings, where there is a single accepted definition, several
incomparable notions of approximate liftings have been proposed.  The first
definition is by \citet{BartheKOZ13} but has some technical shortcomings; we
will use a more recent definition by \citet{BartheO13}. We begin by defining a
distance on distributions, closely related to $(\epsilon, \delta)$-differential
privacy.
\begin{definition}[\citet{BartheO13}]
Let $\epsilon\geq 0$. The $\epsilon$-\emph{DP divergence}
$\Delta_{\epsilon}(\mu_1,\mu_2)$ between two distributions $\mu_1\in\Dist(A)$
and $\mu_2\in\Dist(A)$ is defined as
\[
  \sup_{S\subseteq A} \left(\Pr_{x\leftarrow \mu_1}[x\in S] - \exp(\epsilon)
    \Pr_{x\leftarrow \mu_2} {[x\in S]}\right) .
\]
\end{definition}
For the connection to differential privacy, it is not hard to see that if $M : D
\to \Dist(A)$, then $M$ is $(\epsilon, \delta)$-differentially private iff for
every pair of adjacent inputs $d, d'$, we have $\Delta_\epsilon(M(d), M(d'))
\leq \delta$.  This distance is also central to the definition of \emph{approximate
  liftings}.
\begin{definition}[\citet{BartheO13}] \label{def:approx-lift}
Two distributions $\mu_1\in\Dist(A_1)$ and $\mu_2\in\Dist(A_2)$
are related by the $(\epsilon,\delta)$-\emph{lifting} of
$\Psi\subseteq A_1\times A_2$, written $\mu_1
\alift{\Psi}{(\epsilon,\delta)} \mu_2$, if there exist two \emph{witness}
distributions $\mu_L\in\Dist(A_1\times A_2)$ and
$\mu_R\in\Dist(A_1\times A_2)$ such that
\begin{enumerate}
\item $\pi_1(\mu_L)=\mu_1$ and $\pi_2(\mu_R)=\mu_2$;
\item $\supp(\mu_L)\subseteq \Psi$ and  $\supp(\mu_R)\subseteq \Psi$; and
\item $\Delta_\epsilon(\mu_L,\mu_R)\leq\delta$.
\end{enumerate}
\end{definition}
Approximate liftings generalize several concepts for relating distributions.
When $\mu_L = \mu_R$, we have a $(0, 0)$-lifting, sometimes called an exact
\emph{probabilistic lifting}. Such a lifting, with any $\Psi$, implies a
\emph{probabilistic coupling} between $(\mu_1, \mu_2)$.

Approximate liftings satisfy the following property, also known as the
fundamental lemma of approximate liftings.
\begin{lemma}[\citet{BartheO13}]
Let
$E_1\subseteq B_1$, $E_2\subseteq B_2$, $\mu_1\in\Dist(B_1)$ and
  $\mu_2\in\Dist(B_2)$.  Let
$$\Psi = \{ (x_1,x_2) \in B_1\times B_2 \mid x_1\in E_1 \Rightarrow
  x_2\in E_2 \} .$$
If $\mu_1 \alift{\Psi}{(\epsilon,\delta)} \mu_2$, then
$$\Pr_{x_1\leftarrow\mu_1} [x_1\in E_1] \leq \exp(\epsilon)
\Pr_{x_2\leftarrow\mu_2} [x_2\in E_2] +\delta .$$
\end{lemma}
Using this lemma, one can prove that differential privacy is
equivalent to a particular form of approximate lifting.
\begin{proposition}[\citet{BartheO13}] \label{prop:lift-fund}
  A probabilistic computation $M:D\rightarrow\Dist(A)$ is $(\epsilon,
  \delta)$-differentially private for adjacency relation $\Phi$ iff
  \[
    M(a)~\alift{=}{(\epsilon, \delta)}~M(a')
  \]
  for every two adjacent inputs $a$ and $a'$.
\end{proposition}
Approximate liftings form the basis of the program logic \Saprhl, to which we
turn next.

\subsection{The relational program logic}

The logic \Saprhl, originally proposed by \citet{BartheKOZ13}, is a relational
program logic for verifying differential privacy. We take this logic as our
point of departure; we briefly recall the main points here.

We consider a simple imperative language with random sampling, oracle
calls and adversary calls; the latter two are new to the present work. The
set of commands is defined inductively:
\[
\begin{array}{r@{\ \ }l@{\quad}l}
\Cmd ::= & \Skip                   & \text{noop} \\
     \mid& \Seq{\Cmd}{\Cmd}        & \text{sequencing}\\
     \mid& \Ass{\Var}{\Expr}       & \text{deterministic assignment}\\
     \mid& \Rand{\Var}{\Lap_\epsilon(\Expr)}
                                   & \text{Laplace mechanism}\\
     \mid& \Cond{\Expr}{\Cmd}{\Cmd} & \text{conditional}\\
     \mid& \WWhile{\Expr}{\Cmd}      & \text{while loop} \\
     \mid& \Call{(\Var,\ldots,\Var)}{\mathcal{A}}{\Expr,\ldots,\Expr} &
\text{adversary call} \\
     \mid& \Call{(\Var,\ldots,\Var)}{\mathcal{O}}{\Expr,\ldots,\Expr} &
\text{procedure call}

\end{array}
\]
where $\Var$ is a set of \emph{variables} and $\Expr$ is a set of
\emph{expressions}. Variables and expressions are typed, and range
over standard types like booleans, integers, databases, queries,
lists, etc. We omit the semantics of expressions, which is
standard. Commands are interpreted as maps $\Mem \to \Dist(\Mem)$;
this is also a standard choice (e.g., see
\citet{BartheKOZ13}).\footnote{%
  We will assume that commands are terminating on all executions. The logic
  \Saprhl can also reason about possibly non-terminating programs by working
  with sub-distributions instead of distributions.}
We will write $\dsem{m}{c}$ to mean the output distribution of command $c$,
executed on input memory $m$.

An \Saprhl judgment has the form
\[
  \AEquiv{c}{c'}{\Phi}{\Psi}{\epsilon}{\delta} .
\]
Reminiscent of Hoare logic, $\Phi$ represents the pre-condition while $\Psi$
represents the post-condition. Both $\Phi$ and $\Psi$ are first order formulas
over the program variables. For expressing relational properties, program
variables are tagged with either $\sidel$ or $\sider$ to indicate whether they
belong to $c$ or $c'$ respectively. For instance, we can assert that the
variable $x$ differs by at most $1$ in the two runs with the assertion $|x\sidel
- x\sider| \leq 1$.

Crucially, the post-condition $\Psi$ is interpreted as an approximate lifting
over the output distributions. More formally, the judgment is \emph{valid} iff
for every two memories $m_1$ and $m_2$ such that $m_1 \mathrel{\Phi} m_2$, we have
\[
(\dsem{m_1}{c_1})
\mathrel{\alift{\Psi}{(\epsilon,\delta)}}
(\dsem{m_2}{c_2}) .
\]
We present selected rules, taken from prior presentations of \Saprhl
\citep{BartheKOZ13,BGGHS16} in \cref{fig:aprhl}; $FV(\Phi)$ denotes the set of
program variables in the assertion $\Phi$, and $MV(c)$ denotes the set of
program variables that are modified (i.e., written) by program $c$. Many of the
rules bear a resemblance to the standard Hoare logic rules. The rules
$\textsc{[Assn]}$ and $\textsc{[Cond]}$ are relational versions of the
assignment and conditional rules; note that \textsc{[Cond]} assumes that the two
guards are equal in the pre-condition. The rule \textsc{[Seq]} reflects the
composition principle of approximate liftings, where the indices $\epsilon$ and
$\delta$ add; this rule generalizes the standard composition theorem of
differential privacy. The rule \textsc{[While]} extends this reasoning to loops
with a bound number of iterations, again assuming that the guards are equal in
both programs.

The next two rules, \textsc{[LapNull]} and \textsc{[LapGen]}, are for relating
two sampling instructions from the Laplace distribution. Intuitively \textsc{[LapNull]}
models adding identical noise on both sides, so that the distance between the
samples $(y_1\sidel, y_2\sider)$ is equal to the distance between the means
$(e_1\sidel, e_2\sider)$. \textsc{[LapGen]} is a general rule for assuming that
the two samples are shifted and related by $y_1\sidel + k = y_2\sider$; the
privacy cost depends on how far the means $(e_1\sidel + k, e_2\sider)$ are.

The final group of rules are the structural rules. Besides the usual rules for
consequence and framing (\textsc{[Conseq]} and \textsc{[Frame]}), the most
interesting rule is the \emph{pointwise equality} rule \textsc{[PW-Eq]}. This
rule proves differential privacy by showing a pointwise judgment for each
possible output value $i$, and is the key tool for supporting privacy proofs
beyond the standard composition theorems.

\begin{figure*}[t]
  \[
\begin{array}{c}
  \inferrule*[Left=Assn]
  { }
  { \AEquiv
    {\Ass{x_1}{e_1}}{\Ass{x_2}{e_2}}
    {\Psi\subst{x_1\sidel,x_2\sider}{e_1\sidel,e_2\sider}}
    {\Psi}{0}{0} }
  \\[2ex]
  \inferrule*[Left=Cond]
  { \AEquiv{c_1}{c_2}{\Phi \land b_1\sidel}{\Psi}{\epsilon}{\delta} \\
    \AEquiv{c_1'}{c_2'}{\Phi \land \lnot b_1\sidel}{\Psi}{\epsilon}{\delta} }
  {\AEquiv{\Cond{b_1}{c_1}{c_1'}}{\Cond{b_2}{c_2}{c_2'}}
    {\Phi\land b_1\sidel = b_2\sider}{\Psi}{\epsilon}{\delta}}
  \\[2ex]
  \inferrule*[Left=Seq]
  { \AEquiv{c_1}{c_2}{\Phi}{\Psi'}{\epsilon}{\delta} \\
    \AEquiv{c_1'}{c_2'}{\Psi'}{\Psi}{\epsilon'}{\delta'} }
  { \AEquiv{c_1;c_1'}{c_2;c_2'}{\Phi}{\Psi}{\epsilon + \epsilon'}{\delta + \delta'} }
  \\[2ex]
  \inferrule*[Left=While]
  { \AEquiv
    {c_1}{c_2}
    {\Theta \land b_1\sidel \land b_2\sider \land e\sidel = k}
    {\Theta \land b_1\sidel = b_2\sider \land e\sidel < k}
    {\epsilon_k}{\delta_k}
    \\
    \models \Theta \land e\sidel \leq 0 \to \neg b_1\sidel }
  { \AEquiv
    {\WWhile{b_1}{c_1}}{\WWhile{b_2}{c_2}}
    {\Theta \land b_1\sidel = b_2\sider \land e\sidel \leq n}
    {\Theta \land \neg b_1\sidel \land\neg b_2 \sider}
    {\sum_{k=1}^{n}\epsilon_k}{\sum_{k=1}^{n}\delta_k} }
  \\[2ex]
  \inferrule*[Left=LapNull]
  { y_1 \notin FV(e_1) \qquad y_2 \notin FV(e_2) }
  { \AEquiv {\Rand{y_1}{\Lap_\epsilon(e_1)}} {\Rand{y_2}{\Lap_\epsilon(e_2)}}
    {\top} {y_1\sidel - y_2\sider = e_1\sidel - e_2\sider}{0}{0}}
  \\[2ex]
  \inferrule*[Left=LapGen]
  { }
  { \AEquiv
    {\Rand{y_1}{\Lap_\epsilon(e_1)}}{\Rand{y_2}{\Lap_\epsilon(e_2)}}
    {|k + e_1\sidel - e_2\sider| \leq k'}
    {y_1\sidel+k=y_2\sider}
    {k' \cdot \epsilon}{0}}
  \\[2ex]
  \inferrule*[Left=Conseq]
  { \AEquiv{\Phi'}{c_1}{c_2}{\Psi'}{\epsilon'}{\delta'} \\
    \models \Phi \to \Phi'  \\ \models \Psi' \to \Psi \\ \epsilon' \leq \epsilon \\ \delta' \leq \delta }
  { \AEquiv{c_1}{c_2}{\Phi}{\Psi}{\epsilon}{\delta} }
  \\[2ex]
  \inferrule*[Left=Frame]
  { \AEquiv{c_1}{c_2}{\Phi}{\Psi}{\epsilon}{\delta} \\
    FV(\Theta) \cap MV(c_1, c_2) = \emptyset }
  { \AEquiv{c_1}{c_2}{\Phi \land \Theta}{\Psi \land \Theta}{\epsilon}{\delta} }
  \\[2ex]
  \inferrule*[Left=PW-Eq]
  { \forall i.
    \AEquiv{c_1}{c_2}{\Phi}{x\sidel = i \to x\sider = i}{\epsilon}{\delta_i}}
  {\AEquiv{c_1}{c_2}{\Phi}{x\sidel=x\sider}{\epsilon}{\sum_{i\in I} \delta_i}}
\end{array}
\]
\caption{Selected proof rules of \Saprhl \citep{BartheKOZ13,BGGHS16}}\label{fig:aprhl}
\end{figure*}

\subsection{The union bound logic}

When reasoning about privacy, we will sometimes need to prove probabilistic
bounds on accuracy.  Since accuracy properties are not relational, we cannot
verify them in \Saprhl.  There is a long history of research for formally
verifying probabilistic properties, and we are free to choose any of
these techniques to interface with our logic. In our favor, we are interested in
simple accuracy properties of the form $\Pr[ \Psi ] < \beta$, where $\Psi$ is an
assertion on the program memory. We call such assertions \emph{bad event
  assertions}, since they state that the probability of some event $\Psi$---the
``bad event''---is at most $\beta$. We will prove accuracy assertions in the
Hoare logic \Sahl \citep{BGGHS16-icalp}, which is specialized to prove
bad event assertions.

We will highlight just the features of \Sahl needed for our purposes; readers
should consult \citet{BGGHS16-icalp} for a complete presentation. Concretely,
\Sahl judgments have the following form:
\[
  \ahl{c}{\Phi}{\Psi}{\beta} ,
\]
where $\Phi$ and $\Psi$ are (non-probabilistic) assertions over program
memories, and $\beta \in [0, 1]$ is a real-valued index. Assertions in \Sahl are
non-relational, and mention program variables from a \emph{single} memory
instead of program variables tagged with $\sidel$ or $\sider$. To mediate
between the non-relational assertions of \Sahl and the relational assertions of
\Saprhl, from a non-relational assertion $\Phi$ we can define relational
assertions $\Phi\sidel$ and $\Phi\sider$ by interpreting $\Phi$ where all
program variables are tagged with $\sidel$ or $\sider$ respectively.

The semantics of commands is unchanged from \Saprhl; we interpret commands as
maps $\Mem \to \Dist(\Mem)$. The above judgement means: for any initial
memory satisfying $\Phi$, the probability that $\neg \Psi$ holds in the
resulting distribution on memories is at most $\beta$.  For instance, the
accuracy specification of the Laplace mechanism (\cref{prop:lap-acc}) is given
by the following \Sahl judgment:
\[
  \ahl{\Rand{y}{\Lap_\epsilon(e)}}
  {\top}{ |y - e| \leq \frac{1}{\epsilon} \log \frac{1}{\beta} }{\beta}
\]
for every $\beta \in (0, 1)$.

\section{Accuracy-dependent privacy} \label{sec:accpriv}

Let us begin with our first class of private algorithms, where privacy follows
from an \emph{accuracy} property. For our purposes, these accuracy properties
are non-relational probabilistic properties that hold on a single execution of a
single program. For instance, the assertion $\Pr[ x > 0 ] < 0.2$, stating that
the probability $x$ is positive is at most $0.2$, is an accuracy property.
Accuracy properties appear in privacy proofs in a variety of ways. For instance,
they may imply that the privacy cost $\epsilon$ is smaller than expected. Or,
privacy may be conditional: if the accuracy property holds then we have
differential privacy, otherwise the algorithm fails and there is no guarantee.
Programs in the latter case satisfy $(\epsilon, \delta)$-differential privacy,
where the probability of failure is included in $\delta$.

\subsection{Up-to-bad reasoning}
To integrate accuracy assertions into \Saprhl, we will use a technique from
cryptographic verification: \emph{up-to-bad} reasoning. Roughly speaking, rather
than directly proving the equality lifting corresponding to differential
privacy:
\[
  \mu_1 \alift{(=)}{(\epsilon, \delta)} \mu_2 ,
\]
we will prove a conditional, \emph{up-to-bad} lifting:
\[
  \mu_1 \alift{\{ (x_1, x_2) \mid (\neg \Phi(x_1, x_2) \to x_1 = x_2) \}}{(\epsilon, \delta)} \mu_2 .
\]
Here, $\Phi$ is an assertion involving just variables from one side. Roughly
speaking, the lifting shows that if the \emph{bad event} $\Phi$ does not hold, then we
have differential privacy. Then, we conclude the proof with a structural rule
that combines the bad event assertion---proved externally in \Sahl---with the
up-to-bad lifting, removing the bad event while adjusting the privacy parameters
$(\epsilon, \delta)$.

\begin{figure*}
  \[
    \begin{array}{c}
      \inferrule* [Left=UtB-L]
      { \models \Phi \to \Phi_0\sidel \\
        \AEquiv{c}{c'} { \Phi } { \Theta\sidel \to e\sidel = e\sider}{\epsilon}{\delta} \\
        \ahl{c}{\Phi_0}{\Theta}{\delta'} }
      { \AEquiv{c}{c'} { \Phi } { e\sidel = e\sider }{\epsilon}{\delta + \delta'} }
      \\[2ex]
      \inferrule* [Left=UtB-R]
      { \models \Phi \to \Phi_0\sider \\
       \AEquiv{c}{c'} { \Phi } { \Theta\sider \to e\sidel = e\sider}{\epsilon}{\delta} \\
        \ahl{c'}{\Phi_0}{\Theta}{\delta'} }
      { \AEquiv{c}{c'} { \Phi } { e\sidel = e\sider }{\epsilon}{\delta + e^\epsilon \delta'} }
    \end{array}
  \]
  \caption{Up-to-bad rules}
  \label{fig:other-utb}
\end{figure*}

To support this reasoning in our program logic, we propose the two rules in
\cref{fig:other-utb}.  The rules, \textsc{[UtB-L]} and \textsc{[UtB-R]},
internalize an approximate version of up-to-bad reasoning. If the assertion
$\Theta$ holds, then we have the $(\epsilon, \delta)$-lifting of equality. So,
if we know the probability of $\neg \Theta$ is at most $\delta'$, then we can
show the $(\epsilon, \delta + \delta')$-differential privacy when $\Theta$ is a
property of the first run, and $(\epsilon, \delta + e^\epsilon
\delta')$-differential privacy when $\Theta$ is a property of the second run.
The asymmetry in the left and right versions of the rule reflects the asymmetric
definition of approximate lifting, which is in turn inspired by the asymmetric
definition of differential privacy.

In order to include these rules, we show that they are valid. To prove the
equality lifting for privacy, we would like to use the equivalence in
\cref{prop:lift-fund}. However, there is a catch: we only know that the
distributions over $e$ are differentially private---the distributions over the
whole memory may not be differentially private. Therefore, we will use a new
property of approximate liftings: they are well-behaved when mapping the
underlying distribution.

\begin{restatable}{proposition}{propliftER} \label{prop:lift-ER}
  For a function $f : A \to B$, let $\lift{f} : \Dist(A) \to \Dist(B)$ denote
  function lifted to a map on distributions. If $f$ is surjective, and $R$ is a
  relation on $B$, then
  \[
    \mu_1
    \mathrel{\alift{ \{ (x_1, x_2) \mid f(x_1) \mathrel{R} f(x_2) \}}
    {(\epsilon, \delta)}}
    \mu_2
  \]
  if and only if
  \[
    \lift{f}(\mu_1)
    \mathrel{\alift{ \{ (y_1, y_2) \mid y_1 \mathrel{R} y_2 \}}{(\epsilon,
    \delta)}}
    \lift{f}(\mu_2) .
  \]

  In particular, if we have a set $E$ of equivalence classes of $A$ and the
  distribution $\mu/E : \Dist(E)$ represents the probability of being in each
  equivalence class, taking $f : A \to E$ mapping an element to its equivalence
  class and $R$ to be the equivalence relation gives a result by \citet[Proposition
  8]{BartheO13}:
  \[
    \mu_1 \mathrel{\alift{(=_E)}{(\epsilon, \delta)}} \mu_2
    \iff
    \mu_1/E \mathrel{\alift{(=)}{(\epsilon, \delta)}} \mu_2/E .
  \]
\end{restatable}

This result allows us to prove an approximate lifting for a distribution over
memories by proving an approximting lifting for the distribution over a single
variable or expression. We defer the details of the proof to \thefullversion. Now,
we are ready to show soundness of the up-to-bad rules.

\begin{theorem}
  The rules \textsc{[UtB-L]} and \textsc{[UtB-R]} are sound.
\end{theorem}
\begin{proof}
  We will start with \textsc{[Utb-L]}. Take any two memories $(m_1, m_2)$ such
  that $(m_1, m_2) \models \Phi$, and let $\mu_1, \mu_2$ be $\dsem{m_1}{c}$ and
  $\dsem{m_2}{c'}$ respectively. Note that $m_1 \models \Phi_0$. By validity of
  the premise, we know
  \[
    \Pr_{m \sim \mu_1} [ \neg \Theta ] \leq \delta'
  \]
  and we have a pair of witnesses $\mu_L, \mu_R$ for the relation
  \[
    R \triangleq \Theta\sidel \to e\sidel = e\sider ,
  \]
  such that $\Delta_\epsilon (\mu_L, \mu_R) \leq \delta$.  Our goal is to show
  that the marginal distributions of $\denot{e}$ in $\mu_1, \mu_2$ satisfy
  $(\epsilon, \delta + \delta')$-differential privacy, i.e. for any set $S$,
  \[
    \Pr_{m \sim \mu_1} [ \dsem{m}{e} \in S ]
    \leq e^\epsilon \Pr_{m' \sim \mu_2} [ \dsem{m'}{e} \in S ] + \delta +
    \delta' .
  \]
  To begin, we know that
  \begin{align*}
    \Pr_{m \sim \mu_1} [ \dsem{m}{e} \in S ]
    &= \Pr_{m \sim \mu_1} [ \dsem{m}{e} \in S \land m \models {\Theta} ] \\
    &+ \Pr_{m \sim \mu_1} [ \dsem{m}{e} \in S \land m \models {\neg \Theta}  ] \\
    &\leq \Pr_{m \sim \mu_1} [ \dsem{m}{e} \in S \land m \models {\Theta} ]
    + \delta'
  \end{align*}
  since the probability of $\neg \Theta$ in $\mu_1$ is at most $\delta'$. Now,
  we can conclude with the coupling:
  \begin{align*}
    &\Pr_{m \sim \mu_1} [ \dsem{m}{e} \in S \land m \models {\Theta} ] +
    \delta' \\
    &= \Pr_{(m,m') \sim \mu_L} [ \dsem{m}{e} \in S \land m \models {\Theta} ]
    + \delta' \\
    &\leq e^\epsilon \Pr_{(m,m') \sim \mu_R} [ \dsem{m}{e} \in S \land m \models
    {\Theta} ] + \delta + \delta' \\
    &\leq e^\epsilon \Pr_{(m,m') \sim \mu_R} [ \dsem{m}{e}' \in S ] + \delta + \delta' \\
    &= e^\epsilon \Pr_{m' \sim \mu_2} [ \dsem{m}{e}' \in S ] + \delta + \delta'
    ,
  \end{align*}
  where the first inequality uses $\Delta_\epsilon(\mu_L, \mu_R) \leq \delta$,
  while the second inequality uses $(m, m') \in \supp(\mu_R)$. So, the
  distributions of $\denot{e}$ satisfy differential privacy. By
  \cref{prop:lift-fund} and \cref{prop:lift-ER} with equivalence classes defined
  by the value of $\denot{e}$, we can conclude soundness of \textsc{[UtB-L]}.

  We can show soundness of \textsc{[UtB-R]} in a similar way. Let $\mu_1, \mu_2$
  be as above. We can use the coupling as follows:

  \begin{align*}
    &\Pr_{m \sim \mu_1} [ \dsem{m}{e} \in S ] \\
    &= \Pr_{(m, m') \sim \mu_L} [ \dsem{m}{e} \in S ] \\
    &\leq e^\epsilon \Pr_{(m, m') \sim \mu_R} [ \dsem{m}{e} \in S ] + \delta \\
    &= e^\epsilon \Pr_{(m, m') \sim \mu_R} [ \dsem{m}{e} \in S \land m' \models {\Theta} ] \\
    &+ e^\epsilon \Pr_{(m, m') \sim \mu_R} [ \dsem{m}{e} \in S \land m' \models
    {\neg \Theta} ] + \delta \\
    &\leq e^\epsilon \Pr_{(m, m') \sim \mu_R} [ \dsem{m}{e}' \in S ]
    + e^\epsilon \Pr_{(m, m') \sim \mu_R} [ m' \models {\neg \Theta} ] +
    \delta \\
    &= e^\epsilon \Pr_{m' \sim \mu_2} [\dsem{m'}{e} \in S ]
    + e^\epsilon \Pr_{m' \sim \mu_2} [ m' \models {\neg \Theta} ] + \delta \\
    &\leq e^\epsilon \Pr_{m' \sim \mu_2} [\dsem{m'}{e} \in S ] + \delta + e^\epsilon \delta'
  \end{align*}
  The first inequality uses the bound on the distance between the witnesses:
  $\Delta_\epsilon(\mu_L, \mu_R) \leq \delta$.  The second inequality uses the
  support of $\mu_R$. The final inequality uses the fact that $\Pr_{m' \sim
    \mu_2} [ m' \models {\neg \Theta} ] \leq \delta'$. Since the distributions
  of $\denot{e}$ in $\mu_1, \mu_2$ satisfy $(\epsilon, \delta + e^\epsilon
  \delta')$-differential privacy, we can again use \cref{prop:lift-fund} and
  \cref{prop:lift-ER} to show \textsc{[UtB-R]} is sound.
\end{proof}

\subsection{Propose-Test-Release} \label{sec:other}
To give a small example of up-to-bad reasoning, we can prove privacy for the
\emph{Propose-Test-Release (PTR)} framework~\citep{DworkL09,ThakurtaS13}, a
classic example of privacy depending on an accuracy guarantee.  The goal is to
release the answer to a function $f$. Rather than adding noise directly to the
answer (which may be non-numeric), PTR estimates the \emph{distance to
  instability} of the database $d$ with respect to $f$, denoted
$\mathit{DistToInst}_f(d)$. This quantity measures the distance from $d$ to the
closest database $d'$ such that $f(d) \neq f(d')$, where adjacent databases are
at distance $1$. We will use the following two properties of
$\mathit{DistToInst}$:
\begin{align*}
  & \mathit{DistToInst}_f(d) > 1 \to \forall d'.\; (Adj(d, d') \to f(d) = f(d'))
  \\
  & Adj(d, d') \to |\mathit{DistToInst}_f(d) - \mathit{DistToInst}_f(d')| \leq
  1
\end{align*}

Since the distance itself is private information, PTR first adds noise
to the distance, calling the noisy result $dist$. If it is large
enough, then PTR returns the value of $f$ with no noise. Otherwise,
PTR returns a default value $\bot$. In code:
\[
\begin{array}{l}
  \Rand{dist}{\Lap_{\epsilon}(\mathit{DistToInst}_f(d))}; \\
  \Condt{dist > \ln(1/\delta)/\epsilon + 1}{}\\
  \quad \Ass{r}{f(d)};  \\
  \mathsf{else} \\
  \quad \Ass{r}{\bot}; \\
  \mathsf{return}~r
\end{array}
\]

The key to the privacy of PTR is that if the noise when estimating $dist$ is not
too large, then there are two cases. If $dist$ is large, then there is no
privacy cost for releasing the value of $q$ when $dist$ is large. Otherwise, we
return a default value $\bot$ revealing nothing. In either case, we just have
privacy cost $\epsilon$ from computing the noisy distance.  If the noise when
estimating $dist$ is too large (happening with probability at most $\delta$), we
may violate privacy. So, we get an $(\epsilon, \delta)$-private algorithm.

\begin{theorem}
  PTR is $(\epsilon, \delta)$-differentially private for $\epsilon, \delta > 0$.
\end{theorem}
\begin{proof}[Proof sketch]
  To prove differential privacy, we will use the rule \textsc{[UtB-L]}. We can
  take the event $\Theta$ to hold exactly when the noise is not too large:
  \[
    \Theta \triangleq |dist - \mathit{DistToInst}_f(d)| < \ln(1/\delta)/\epsilon
    .
  \]
  Then, we couple the noise to be the same, so that we take the same branch in
  both runs.  Then, under assumption $\Theta$, we know that if we estimate that the
  distance to instability is large and we take the first branch, then in fact
  $f(d\sidel) = f(d\sider)$ as desired. Finally, using a tail bound for the
  Laplace distribution in \Sahl gives
  \[
    \ahl{c}{\Phi\sidel}{\Theta}{\delta} ,
  \]
  and rule \textsc{[UtB-L]} gives $(\epsilon, \delta)$-differential privacy.
\end{proof}

\section{Advanced composition} \label{sec:advcomp}
Advanced composition is a key tool for giving a more precise privacy
analysis of a composition of private programs. In this section, we extend
\Saprhl with a new loop rule that supports advanced composition for arbitrary
invariants and we show how the rule can be applied for proving privacy of a
non-interactive variant of \EXAMPLE.

Before presenting our solution, we stress that there are some technical
challenges in extending \Saprhl with advanced composition. On the one
hand, \Saprhl derives its expressiveness from its ability to reason
about \emph{arbitrary} approximate liftings, and not simply about
approximate liftings for the equality relation; as a consequence, an
advanced composition rule for \Saprhl should support
approximate liftings in order to be useful (in fact,
an advanced composition rule for just approximate liftings
of equality would be too weak for verifying our running example). On
the other hand, the proof of advanced composition is substantially more
technical than the proof for the sequential composition theorem, using
sophisticated results from fields such as martingale
theory and hypothesis testing. We overcome these obstacles by showing
approximate liftings from a version of differential privacy; as we saw before,
differential privacy is also a consequence of an approximate lifting of
equality.  The key observation is that the two witnesses $\mu_L$ and $\mu_R$
used in the definition of an approximate lifting define a mechanism
$\mu:\mathbb{B}\rightarrow\Dist (A_1\times A_2)$ such that
$\mu(\mathsf{true})=\mu_L$ and $\mu(\mathsf{false})=\mu_R$ where $\Delta_\epsilon
(\mu_L,\mu_R)\leq\delta$ iff $\mu$ is $(\epsilon,\delta)$-differentially
private. Next, we show how to take advantage of this observation.

Since the $(\epsilon, \delta)$ parameters from approximate lifting are defined
by the $\epsilon$-distance $\Delta_\epsilon(\mu_1, \mu_2)$ between the two
witnesses, we will first show an advanced composition theorem for
$\epsilon$-distance.


\begin{proposition}[Advanced composition for $\epsilon$-distance]\label{prop:advcompdist}
Let $f_i,g_i: A \to \Dist(A)$ such that $\Delta_\epsilon(f_i(a), g_i(a))\leq
\delta$ for every $a\in A$. For any $\omega > 0$, let:
\[
  \epsilon^* \triangleq \left(\sqrt{2 n \ln(1/\omega)}\right) \epsilon
    + n \epsilon(e^\epsilon - 1)
  \quad \text{and} \quad
  \delta^* \triangleq n \delta + \omega .
\]
Then for every $n\in\mathbb{N}$ and $a\in A$,
$\Delta_{\epsilon^*}(f^n(a), g^n(a))\leq \delta^*$.
\end{proposition}
\begin{proof}
Let $h_i: A \to \mathbb{B} \to \Dist(A)$ be such that for every $a\in
A$, $h_i(a, \mathsf{true})=f_i(a)$ and $h_i(a, \mathsf{false})=g_i(a)$.
Then $\Delta_\epsilon(f_i(a), g_i(a))\leq \delta$ iff $h_i(a) :
\mathbb{B} \to \Dist(A)$ is $(\epsilon, \delta)$-differentially
private for every $a \in A$.

By directly applying the advanced composition theorem of differential privacy
(\cref{thm:advcomp-dp}), the function $h^n(a): \mathbb{B} \to \Dist(A)$
is $(\epsilon^*,\delta^*)$-differentially private for each $a \in A$. So for every
$b,b'\in\mathbb{B}$ and $a\in A$, $\Delta_{\epsilon^*} (h^n(a,b),h^n(a, b'))\leq
\delta^*$. Now for every $a\in A$, $h^{n}(a, \mathsf{true})=f^{n}(a)$ and
$h^{n}(a, \mathsf{false})=g^{n}(a)$. Therefore, $\Delta_{\epsilon^*}
(f^n(a),g^n(a))\leq \delta^*$.
\end{proof}

Now that we have an advanced composition for $\epsilon$-distance, it is a simple
matter to extend our result to approximate liftings. Note here that we apply
advanced composition not to the distributions on $A$---which are related by an
approximate lifting, but perhaps \emph{not} related by differential privacy---but
rather to the two \emph{witnesses} of the lifting, distributions on pairs in $A
\times A$.

\begin{proposition}[Advanced composition for lifting]\label{prop:advcomplift}
Let $f_i,f_i':A\rightarrow\Dist(A)$ and $\Phi\subseteq A\times A$ such that
for every $a,a'\in A$, $(a,a')\models\Phi$ implies
\[
  f_i(a) \mathrel{\alift{\Phi}{(\epsilon, \delta)}} f_i'(a') .
\]
Let $n\in\mathbb{N}$ and let $(\epsilon', \delta')$ be as in
\cref{thm:advcomp-dp}.  For any $\omega > 0$, let
\[
  \epsilon' \triangleq \left(\sqrt{2 n \ln(1/\omega)}\right) \epsilon
    + n \epsilon(e^\epsilon - 1)
  \quad \text{and} \quad
  \delta' \triangleq n \delta + \omega .
\]
Then for every $a,a'\in A$ such that $(a,a')\models\Phi$, we have
\[
  f^n(a) \mathrel{\alift{\Phi}{(\epsilon', \delta')}} f'^n(a') .
\]
\end{proposition}
\begin{proof}
We can map any pair $(a, a') \in \Phi$ to the left and right witnesses of the
approximate lifting. That is, there exists $h^l_i,h^r_i:(A\times A) \rightarrow
\Dist (A\times A)$ such that for every $(a,a')\models\Phi$:
\begin{itemize}
\item $\pi_1 (h^l_i(a,a'))=f_i(a)$ and $\pi_2 (h^r_i(a,a'))=f_i'(a')$
\item $\supp(h^l_i(a,a'))\subseteq \Phi$ and $\supp(h^r_i(a,a'))\subseteq
  \Phi$
\item $\Delta_{\epsilon} (h^l_i(a,a'),h^r_i(a,a')) \leq \delta$
\end{itemize}
Without loss of generality, we can assume that $h^l_i(a,a')=h^r_i(a,a')=0$ if
$(a,a')\models \neg\Phi$. By \cref{prop:advcompdist}, we also have
$\Delta_{\epsilon'}(h^l_i(a,a'), h^r_i(a,a'))\leq \delta'$ for every
$i\in\mathbb{N}$ and $(a,a')\models \Phi$. By induction on $i$, for every
$(a,a')\models \Phi$ we have $\supp((h^l)^n(a,a')) \subseteq \Phi$ and
$\supp((h^r)^n(a,a'))\subseteq \Phi$, $\pi_1 ((h^l)^n(a,a'))=f^n(a)$ and $\pi_2
((h^r)^n(a,a'))=f_i'^n(a')$.
\end{proof}

\begin{figure*}
  \[
    \inferrule*[Left=AC-While]
    { \models \Theta \land e\sidel \leq 0 \to \neg b_1\sidel \and
      \epsilon^* \triangleq
        \left(\sqrt{2 n \ln(1/\omega)}\right) \epsilon +
          n \epsilon(e^\epsilon - 1) \and
      \delta^* \triangleq n \delta + \omega \\
       \omega > 0 \\\\
      \AEquiv{c_1}{c_2}
      {\Theta \land b_1\sidel \land b_2\sider \land e\sidel = k}
      {\Theta \land b_1\sidel = b_2\sider \land e\sidel < k}
      {\epsilon}{\delta}
    }
    { \AEquiv
      {\WWhile{b_1}{c_1}}{\WWhile{b_2}{c_2}}
      {\Theta \land b_1\sidel = b_2\sider \land e\sidel \leq n}
      {\Theta \land \neg b_1\sidel \land\neg b_2 \sider}
      {\epsilon^*}{\delta^*} }
  \]
  \caption{Advanced composition rule}\label{fig:adv:comp:while}
\end{figure*}

We remark that our connection between the witnesses of liftings and differential
privacy allows us to directly import other composition theorems of differential
privacy and their proofs without change. For instance,
\citet{DBLP:journals/corr/OhV13} consider two
variants of advanced composition: an optimal variant that provably gives the
best bound on $\epsilon$ and $\delta$, and a heterogeneous variant that allows
$\epsilon$ and $\delta$ be different for the different mechanisms. In
unpublished work, \citet{RRUV16} consider a version of the advanced composition
theorem where the privacy level $\epsilon_i$ and $\delta_i$ for the $i$-th
mechanism may be chosen \emph{adaptively}, i.e., depending on the results from
the first $i - 1$ mechanisms.  These composition theorems are quite tricky to
prove, involving sophisticated tools from martingale theory and hypothesis
testing. We expect that we can internalize all of these composition
theorems---and directly generalize to liftings---with minimal effort.

Based on the previous result, we introduce a new rule
\textsc{[AC-While]} that formalizes advanced composition for
loops.
The soundness for the new rule, which is given in
\cref{fig:adv:comp:while} follows immediately from the results of the
previous section.
\begin{theorem}
  The rule \textsc{[AC-While]}  is sound.
\end{theorem}

\section{Interactive privacy} \label{sec:interactive}

So far, we have seen how to incorporate composition theorems and accuracy proofs
into our logic. Now, we consider the last piece needed to verify \EXAMPLE:
proving privacy for \emph{interactive} algorithms. To date, privacy has only
been formally verified for algorithms where the entire input is available in a
single piece; such algorithms are called \emph{offline} algorithms. In contrast,
\emph{interactive} or \emph{online} algorithms accept input piece by piece, in a
finite stream of input, and must produce an intermediate outputs as inputs
arrive.

The differential literature proposes several interactive algorithms; examples
include private counters \citep{DNPR10,CSS10}, the Sparse Vector mechanism, and other
algorithms using these mechanisms \citep{HR10}. The main difficulty in verifying
privacy is to model \emph{adaptivity}: later inputs can depend on
earlier outputs. Indeed, differential privacy behaves well under
adaptivity, a highly useful property enabling applications to adaptive
data analysis and statistics \citep{dwork2015reusable}.

We can view adaptive inputs as controlled by an \emph{adversary}, who
receives earlier outputs and selects the next input. We draw on
techniques for formally verifying cryptographic proofs, which often
also involve an adversary who is trying to break the protocol. We take
inspiration from the treatment of adversaries in the logic \Sprhl, an
exact version of \Saprhl that has been used for verifying
cryptographic proofs~\citep{BartheGZ09}. Specifically, we extend
\Saprhl with a rule \textsc{[Adv]} for the adversary. The rule,
displayed in \Cref{fig:adv-rules}, generalizes the adversary rule from
\Sprhl; let $\Phi$ an assertion that does not contain any adversary
variable, and assume that the adversary $\A$ has access to oracles
$\Or_1$, \ldots, $\Or_n$ and that each oracle guarantees equality of
outputs and an invariant $\Phi$, provided it is called on equal inputs
that satisfy $\Phi$.  Then, $\A$ guarantees equality of outputs and an
invariant $\Phi$, provided it is called on equal inputs that satisfy
$\Phi$. Moreover, the privacy cost of calling the adversary $\A$ is
equal to $\langle\sum_{k=1}^{n}q_k \epsilon_k,
\sum_{k=1}^{n}q_k\delta_k\rangle$ where $\langle
\epsilon_i,\delta_i\rangle$ is the cost of calling once the oracle
$\Or_i$, and $q_i$ is the maximal number of adversarial queries for
oracle $\Or_i$.
\begin{figure*}
$$\begin{array}{c}
\inferrule*[Left=Adv]
    {\forall i, \vec{y}, \vec{z}.\,
        \AEquiv{\Ass{\vec{y}}{\Or_i(\vec{z})}}{\Ass{\vec{y}}{\Or_i(\vec{z})}}
              {\vec{z}\sidel = \vec{z}\sider \land \Phi }
              {\vec{y}\sidel = \vec{y}\sider \land \Phi}
              {\epsilon_i}{\delta_i}
    }
    { \AEquiv{\Ass{\vec{x}}{\A(\vec{e})}}{\Ass{\vec{x}}{\A(\vec{e})}}
            {\vec{e}\sidel = \vec{e}\sider \land \mathbf{x}_{\mathcal{A}}
            \sidel =  \mathbf{x}_{\mathcal{A}}\sider \land  \Phi }
            {\vec{x}\sidel = \vec{x}\sider \land
   \mathbf{x}_{\mathcal{A}}\sidel =  \mathbf{x}_{\mathcal{A}}\sider  \land \Phi}
            {\sum_{k=1}^{n}q_k\epsilon_k}{\sum_{k=1}^{n}q_k\delta_k}}
\end{array}$$
where $q_1, \ldots, q_n$ are the maximal number of queries that $\A$ can make to
oracles $\Or_1, \ldots, \Or_n$ and $\mathbf{x}_\A$ is the state of the adversary.
\caption{Adversary rule}
\label{fig:adv-rules}
\end{figure*}
One can prove the soundness of the adversary rule by induction on the
code of the adversary.
\begin{proposition}
The rule \textsc{[Adv]} is sound.
\end{proposition}
Note that the proof of sparse vector only makes a restricted use of
the \textsc{[Adv]} rule: as $\A$ does not have access to any oracle,
the \Sprhl\ rule suffices for the proof. However, the following,
oracle based, presentation of sparse vector uses the full power of the
\textsc{[Adv]} rule:
\[
  \begin{array}{l}  
    \Ass{l}{[]}; \\  
    \Rand{u}{\Lap_{\epsilon/2}(0)}; \\
    \Ass{A}{a-u}; \\
    \Ass{B}{b+u}; \\
    \Call{x}{\mathcal{A}^\mathcal{O}}{}; \\
    \mathsf{return}~l 
  \end{array}
\]
where $\mathcal{O}$ is an oracle that takes a query and checks whether
it is between thresholds and updates a public list $l$, and $\A$ is
allowed to query $\mathcal{O}$ up to $N$ times. In addition, we note
that our new rule can also be useful for cryptographic proofs which
involve reasoning about statistical distance.

\section{Optimal subset coupling} \label{sec:optimal}
\begin{figure*}
  \[
  \inferrule* [Left=LapInt]
  { \epsilon' \triangleq \ln \left( \frac{\exp( \eta \epsilon )}{1 - \exp(-\sigma\epsilon/2)} \right) \\
    \Phi \triangleq |e\sidel - e\sider| \leq k \land
    (p + k \leq r < s \leq q - k) \land
    (q - p) - (s - r) \leq \eta \land
    0 < \sigma \leq (s - r) + 2
  }
  { \AEquiv{\Rand{y_1}{\Lap_\epsilon(e)}}{\Rand{y_2}{\Lap_\epsilon(e)}}
    { \Phi }{ y_1\sidel \in [p, q] \leftrightarrow y_2\sider \in [r, s] } { \epsilon' }{ 0 } }
\]
\caption{Interval coupling for Laplace}
\label{fig:lap-int}
\end{figure*}

The privacy proof of \EXAMPLE relies on a new interval coupling rule
\textsc{[LapInt]} for Laplace sampling, \cref{fig:lap-int}. This rule allows us
to relate a larger interval with a smaller interval nested inside. That is, we
can assume that the sample $y_1\sidel$ lies in $[p, q]$ if and only if the
sample $y_2\sider$ lies in $[r, s]$ contained in $[p, q]$. The privacy cost
depends on two things: the difference in sizes of the two intervals
$(q - p) - (s - r)$, and the size of the inner interval $s - r$. Roughly, a
larger inner interval and smaller outer interval yield a smaller privacy cost.

To show that this rule is sound, we will first prove a general construction for
liftings of the form
\[
  \mu \alift{(y_1 \in P \leftrightarrow y_2 \in Q)}{(\epsilon, 0)} \mu
\]
for $Q \subseteq P$. We call such such liftings \emph{subset couplings}, since
they relate a set of outputs to a subset of outputs. Our construction applies to
all discrete distributions, and is \emph{optimal} in a precise sense: among all
liftings of the relation, our construction gives the smallest (i.e., the most
precise) $\epsilon$.

\begin{theorem}[Optimal subset coupling] \label{thm:opt-subset}
  Let $\mu$ be a distribution over $S$, and consider two proper
  subsets $P, Q$ of $S$ such that $Q \subseteq P$. Then $\mu(P) \leq
  \alpha \mu(Q)$ if and only if the following lifting holds:
  \[
    \mu \alift{R}{(\ln \alpha, 0)} \mu
  \]
where the relation $R$ is defined by the following clause:
  \[
  (y_1,y_2)\in R \triangleq  y_1 \in P \leftrightarrow y_2 \in Q
  \]
\end{theorem}
\begin{proof}
  The reverse direction follows from the fundamental lemma of approximate
  liftings~\citep{BartheO13}. For the forward implication, we construct two
  witnesses. Note that the theorem is trivial if $\mu(P \setminus Q) = 0$ since
  we can just take the identity coupling, so we will assume otherwise. Define
  the following witnesses:
  \begin{align*}
    \mu_L(x, y) &=
    \begin{cases}
      \mu(x)
      & \text{if } x \notin P\setminus Q \land x = y \\
      \frac{\mu(x) \mu(y)}{\mu(Q)}
      & \text{if } x \in P \setminus Q \land y \in Q \\
      0    & \text{otherwise} .
    \end{cases}
    \\
    \mu_R(x, y) &=
    \begin{cases}
      \mu(y)
      & \text{if } x = y \land y \notin  P \\
      \lambda \cdot \mu(y)
      & \text{if } x = y \land y \in  Q \\
      \frac{(1 - \lambda) \mu(x) \mu(y)}{\mu(P \setminus Q)}
      & \text{if } x \in P \setminus Q \land y \in Q \\
      \mu(y)
      & \text{if } x = x_0 \in S \setminus P \land y \in P \setminus Q \\
      0    & \text{otherwise} .
    \end{cases}
  \end{align*}
  Here, $x_0$ is an arbitrary element of $S \setminus P$. We set $\lambda =
  \mu(Q) / \mu(P)$. Note that $\lambda$ satisfies:
  \[
    \lambda = (1 - \lambda) \frac{\mu(Q)}{\mu(P \setminus Q)} .
  \]
  For the witnesses, it is not hard to see that the marginal
  conditions are satisfied, and that
  $x \mathrel{R} y$ for all pairs $(x, y)$ in the supports of $\mu_L$ and $\mu_R$.
  Furthermore, for all $x\in P$ and $y\in Q$, we have $\mu_L(x, y) =
  (1/\lambda) \cdot \mu_R(x, y)$ by our choice of $\lambda$. By the
  condition on marginals and the support, we have a $(\ln(1/\lambda),
  0)$-lifting.  This immediately implies the $(\ln \alpha, 0)$-lifting
  for any $\alpha \geq 1/\lambda = \mu(P)/\mu(Q)$.
\end{proof}

This result is very much in the spirit of the optimal or \emph{maximal} coupling
for exact probabilistic couplings (see, e.g., \citep{Thorisson00}). By computing
the total variation distance between two distributions, the maximal coupling
construction shows how to create a coupling that exactly realizes the total
variation distance.

Similarly, by the optimal subset coupling, we can construct a subset coupling
and calculate the distance $\epsilon$ in the lifting by simply proving a
property of the discrete Laplace distribution.  Formally, let $\Lap_\epsilon(v)$
for $v \in \mathbb{Z}$ have distribution over $\mathbb{Z}$ with probability
proportional to
\[
  \Pr[r] \propto \exp(- \epsilon |r - v|) .
\]
We write $\Lap_\epsilon$ to mean $\Lap_\epsilon(0)$.  We will use the following
property, a discrete version of \citet[Claim 5.13]{BunSU16}.

\begin{restatable}{lemma}{lemlapint} \label{lem:lap-int}
  Let $r$ be a draw from $\Lap_\epsilon$, and take $a, a', b, b' \in \mathbb{Z}$
  such that $a < b$ and $[a, b] \subseteq [a', b']$.  Then,
  \[
    \Pr[ r \in [a', b'] ] \leq \alpha \Pr[ r \in [a, b] ]
  \]
  with
  \[
    \alpha = \frac{\exp( \eta \epsilon )}{1 - \exp(-(b - a + 2)\epsilon/2)},
    \qquad
    \eta = (b' - a') - (b - a) .
  \]
\end{restatable}
The proof follows by a small calculation; we defer the details to
\thefullversion.  With this property, we can now prove soundness of our subset
rule for sampling from the Laplace distribution.

\begin{theorem}
  The rule \textsc{[LapInt]} is sound.
\end{theorem}
\begin{proof}
  Suppose that $\denot{e\sidel} = v_1, \denot{e\sider} = v_2$, and $|v_1 - v_2|
  = \Delta \leq k$. Let the noises be $w_1 = x\sidel - v_1, w_2 = x\sider =
  v_2$; note that both samples are distributed as $\Lap_\epsilon(0)$. Note that
  $x\sidel \in [a', b']$ exactly when $w_1 \in [a' - v_1, b' - v_1]$, and
  similarly for $x_2$ and $w_2$. By \cref{prop:lift-ER}, to show a lifting on
  memories, it suffices to find a lifting on the distribution over the sampled
  variable $y$, taking $f$ to be the function that maps a memory $m$ to the
  value $m(y)$. So, it suffices to show
  \[
    \Lap_\epsilon(0)
    \alift{ \{ w_1 \in I_1 \leftrightarrow w_2 \in I_2 \} }{(\epsilon', 0)}
    \Lap_\epsilon(0)
  \]
  where $I_1 \triangleq [a' - v_1, b' - v_1]$ and $I_2 \triangleq [a -
    v_2, b - v_2]$. Since $a\sidel + k \leq a\sider$ and $b\sidel \leq
  b\sider - k$, we know $I_2 \subseteq I_1$. Then, we can directly
  apply \cref{lem:lap-int} on these two intervals with $\eta = |I_1| -
  |I_2| = (b' - a') - (b - a)$, and we are done.
\end{proof}

\section{Proving privacy for \EXAMPLE} \label{sec:bt}
We verify differential privacy for the full version of \EXAMPLE with
an adaptive adversary that chooses its queries interactively. We recall the
theorem.

\begin{theorem}
  Let $\epsilon$ and $\delta$ both be in $(0, 1)$. Set
  \[
    \epsilon' \triangleq \frac{\epsilon}{4 \sqrt{2 M \ln(2/\delta)}} .
  \]
  If all adversarial queries $q$ are $1$-sensitive (i.e.\,
  $|\mathsf{evalQ}(q,d)-\mathsf{evalQ}(q,d')| \leq 1$ for every adjacent
  databases $d$ and $d'$), then \EXAMPLE is $(\epsilon,\delta)$-differentially
  private. Formally,
  \begin{equation*}
    \AEquiv{\mathsf{ASV}_{\mathrm{bt}}}{\mathsf{ASV}_{\mathrm{bt}}}
    {\Phi}{r\sidel = r\sider}
    {\epsilon}{\delta}
  \end{equation*}
  where $\Phi$ is defined as
  \[
    =_{\{a,b,N,qs\}}\land \mathsf{Adj}(d\sidel,d\sider) \land b\sidel - a\sidel
    \geq \gamma
  \]
  for
  \[
    \gamma \triangleq \frac{6}{\epsilon'} \ln (4/\epsilon') + \frac{4}{\epsilon}
    \ln(2/\delta) .
  \]
\end{theorem}
\begin{proof}[Proof sketch]
Now, we put everything together to prove \EXAMPLE is $(\epsilon,
\delta)$-private algorithm. We present just the key points of the proof here;
the full proof is formalized in the \EasyCrypt system. We first transform the
algorithm to an equivalent algorithm:\footnote{%
  We have formally verified equivalence of this program with the program in
  \cref{fig:running} by using a recently-proposed asynchronous loop
  rule~\citep{BGHS16}.}

\[
  \begin{array}{l} \\
  \mbox{\EXAMPLE}(a,b,M,N,d) := \\
    \Ass{i}{0};
    \Ass{l}{[]}; \\
    \Rand{u}{\Lap_{\epsilon/2}(0)}; \\
    \Ass{A}{a-u};
    \Ass{B}{b+u}; \\
    \WWhile{i < N \land |l| < M}{} \\
    \quad \Ass{i'} i; \Ass{hd} -1; \\
    \quad \WWhile{i' < N}{} \\
    \quad\quad  \mathsf{if}~ (hd = -1)~ \\
    \quad\quad\quad  \Call{q}{\mathcal{A}}{l};\\
    \quad\quad\quad  \Rand{S}{\Lap_{\epsilon'/3}(\mathsf{evalQ}(q,d))};  \\
    \quad\quad\quad  \mathsf{if}~ (A\leq S \leq B)~ \mathsf{then}~\Ass{hd}{i}; \\
    \quad\quad\quad  \Ass{i}{i + 1}; \\
    \quad\quad  \Ass{i'}{i' + 1}; \\
    \quad  \mathsf{if}~ (hd \neq -1)~\mathsf{then}~\Ass{l}{hd :: l}; \\
    \mathsf{return}~l
  \end{array}
\]
The main change is that the loop iterations in \cref{fig:running} are grouped
into blocks of queries, each handled by an inner loop. Each outer iteration now
corresponds to finding a single query between the thresholds. The inner
iterations loop through the queries until there is a between threshold query. To
allow the inner loop to be analyzed in a synchronized fashion, the inner loop
always continues up to iteration $N$, doing nothing for all iterations beyond
the first between threshold query.

This transformation allows us to use both advanced composition and pointwise
equality. At a high level, we follow four steps. First, we set up the threshold
noise, so that the noisy interval $[A, B]$ is smaller in the first run than in the
second run; this costs a bit of privacy. Then, we handle the loop working inside
to out: we apply the adversary rule, the subset coupling, and pointwise equality
to the inner loop to show privacy for each block of queries that stops as soon
as we see a between threshold query, assuming that the noisy intervals $[A, B]$
are sufficiently large. Next, we apply advanced composition to bound the privacy
cost of the outer loop, still assuming $[A, B]$ is sufficiently large. Finally,
we use up-to-bad reasoning to remove this assumption, increasing $\delta$
slightly for the final $(\epsilon, \delta)$-privacy bound.

We now detail each step. To reduce notation, we will suppress the adjacency
predicate $Adj(d\sidel, d\sider)$ which is preserved throughout the computation
and the proof.

\paragraph*{Threshold coupling}
Let $c_t$ be the initialization command, including all commands before the loop.
First, we can prove:
\begin{equation*}
  \AEquiv{c_t}{c_t}{\Phi}{\Phi'}{\frac{\epsilon}{2}}{0}
\end{equation*}
where $\Phi' \triangleq A\sidel + 1 = A\sider \land B\sidel = B\sider + 1$,
by applying the rule \textsc{[LapGen]} to ensure $u\sidel = u\sider + 1$ as a
post-condition. This step costs $(\epsilon/2, 0)$ privacy.

\paragraph*{The inner loop}

For the inner loop, we will assume the threshold condition $\Phi'$ and $l\sidel
= l\sider$ initially; both conditions are preserved by the inner loop. We will
also assume that the noisy interval $[A, B]$ is sufficiently large:\footnote{%
  While $\Psi$ does not have tagged variables, we will later interpret $A$ and
  $B$ as coming from the first run.}
\[
  \Psi \triangleq B - A \geq \frac{6}{\epsilon'} \ln (4/\epsilon') .
\]
Let $c_i$ be the inner loop. We will first prove the pointwise judgment:
\begin{equation} \label{eq:bt-pw}
  \begin{split}
    \vdash &{c_i} \sim_{\!\left\langle \epsilon',0\right\rangle} {c_i} : \\
    &{\Phi' \land \Psi\sidel \land l\sidel = l\sider} \Longrightarrow {(hd\sidel = v) \to (hd\sider = v)}
  \end{split}
\end{equation}
for each index $v$.

We focus on the case where $0 \leq v \leq N$, as other cases are easy.  First, we
apply the \textsc{[While]} rule, with $\epsilon_i =0$ except for $i=v$, where we
set $\epsilon_v=\frac{\epsilon}{2}$. Whenever we call the adversary for the next
query, we know that $l\sidel = l\sider$, so we may apply adversary rule with
cost $(0, 0)$ to ensure that $q\sidel = q\sider$ throughout.

Then, the proof for the rest of the loop body goes as follows:
\begin{itemize}
  \item For the iterations $i < v$ and $i > v$, we use the rule
    \textsc{[LapNull]} to couple the noisy query answers $S\sidel, S\sider$.
    This has no privacy cost and preserves the invariant.

  \item For the iteration $i = v$, suppose that $S\sidel \in
    [A\sidel, B\sidel]$ (otherwise we are done). We can apply a
    coupling for the Laplace distribution, \textsc{[LapInt]},
    to ensure that $S\sider \in [A\sider, B\sider]$
    as well.  Under $\Psi\sidel$ and the coupling on the noisy
    thresholds $\Phi'$, the inner interval $[A\sider, B\sider]$ has
    size at least $(6 / \epsilon') \ln (4 / \epsilon') - 2$. Taking $(p, q,
    r, s) = (A\sidel, B\sidel, A\sider, B\sider)$, $\eta = 2$,
    $\sigma = (6 / \epsilon') \ln (4 / \epsilon')$, and $k = 1$, a
    calculation\iffull\footnote{%
      We will show that for $\lambda \in (0, 1/2)$, we have
      \[
        \frac{\exp(2 \lambda/3)}{1 - \exp(- \sigma \lambda/6)} \leq \exp(\lambda)
      \]
      when $\sigma \geq (6/\lambda) \ln (4/\lambda)$.  Substituting, it suffices
      to show:
      \begin{align*}
        \frac{\exp(2 \lambda/3)}{1 - \lambda/4} &\leq \exp(\lambda) \\
        \lambda/4 + \exp(-\lambda/3) - 1 &\leq 0 .
      \end{align*}
      Since the function on the left is convex in $\lambda$, the maximum
      occurs on the boundary of the domain. We can directly check that the
      inequality holds at $\lambda = \{ 0, 1/2 \}$, and we are done.
    }\fi{}
    shows that \textsc{[LapInt]} gives a $(\epsilon', 0)$-lifting so the
    critical iteration has privacy cost $\epsilon'$.
\end{itemize}
This establishes \cref{eq:bt-pw}. By the pointwise equality rule
\textsc{[PW-Eq]}, we have:
\[
  \AEquiv{c_i}{c_i}{\Phi' \land \Psi\sidel \land l\sidel = l\sider}{hd\sidel = hd\sider}{\epsilon'}{0}
\]
By \textsc{[Frame]} and some manipulations, we can assume that $l\sidel =
l\sider$ at the end of each iteration of the outer loop.

\paragraph*{The outer loop}
For the outer loop, we apply advanced composition. Letting $c_o$ be the outer
loop, our choice of $\epsilon'$ and corresponds to the setting in
\cref{thm:advcomp-dp}, so we have the following judgment by \textsc{[AC-While]}:
\[
  \AEquiv{c_o}{c_o}{l\sidel = l\sider \land \Phi' \land \Psi\sidel}
  {l\sidel = l\sider}{\epsilon/2}{0} .
\]
Since $c_o$ does not modify the thresholds and preserves $\Psi\sidel$,
\textsc{[Frame]} and some manipulations allows us to move this assertion into
the post-condition:
\[
  \AEquiv{c_o}{c_o}{l\sidel = l\sider \land \Phi'}
  {\Psi\sidel \to l\sidel = l\sider}{\epsilon/2}{0} .
\]

\paragraph*{Applying up-to-bad reasoning}
Finally, we can apply \textsc{[Seq]} with our judgement for the initialization
$c_i$ and the outer loop $c_o$, giving:
\[
  \AEquiv{\mbox{\EXAMPLE}}{\mbox{\EXAMPLE}}{\Phi_0}
  {\Psi\sidel \to l\sidel = l\sider}{\epsilon/2}{0}
\]
for
\[
  \Phi_0 \triangleq~ =_{\{a,b,N,qs\}}\land\ \mathsf{Adj}(d\sidel,d\sider)
                     \land b\sidel - a\sidel \geq \gamma .
\]
To conclude the proof, all that remains is to remove the assertion $\Psi\sidel$.
We will bound the probability that $\Psi\sidel$ does not hold. The accuracy rule
for the Laplace mechanism gives
\[
  \ahl{\Rand{u}{\Lap_{\epsilon/2}(0)}}{b - a \geq \gamma}
  {|u| \leq \frac{2}{\epsilon}\log(1/\delta)}{\delta} ,
\]
from which we can conclude
\[
  \ahl{\mbox{\EXAMPLE}}{b - a \geq \gamma}{\Psi}{\delta} .
\]
Finally, applying \textsc{[UtB-L]} yields the desired judgment:
\[
  \AEquiv{\mbox{\EXAMPLE}}{\mbox{\EXAMPLE}}{\Phi_0}
  {l\sidel = l\sider}{\epsilon/2}{\delta}
  \qedhere
\]
\end{proof}

\section{Related Works} \label{sec:related}

Differential privacy~\citep{DMNS06} has been an area of intensive research in
the last decade. We refer readers
interested in a more comprehensive treatment of the algorithmic
aspects of differential privacy to the excellent monograph by~\citet{DR14}.
Several tools have been developed to support the development of
differentially private data analysis.
PINQ \citep{pinq} internalizes the use of standard composition in the form
of a privacy budget management platform, Airavat \citep{RSKSW10} uses
differential privacy combined with the map-reduce approach, GUPT \citep{MTSSC12}
implements the general idea of sample and
aggregate \citep{NRS07}. Other tools implement
algorithms targeting specific applications like: location data \citep{OnTheMap},
genomic data \citep{FienbergSU11,SSB16}, mobility data \citep{DP-where},
and browser error reports \citep{EPK14}.

Several tools have been proposed for providing formal verification of
the differential privacy guarantee, using a wide variety of verification
approaches:
dynamic
checking~\citep{pinq,conf/popl/EbadiSS15}, relational program
logic~\citep{BartheKOZ13,BartheO13} and relational refinement type
systems~\citep{BGGHRS15}, linear (dependent) type
systems~\citep{ReedPierce10,GHHNP13}, product
programs~\citep{BGGHKS14}, methods based on computing bisimulations
families for probabilistic
automata~\citep{Tschantz201161,xu:hal-00879140}, and  methods based on
counting variants of satisfiability modulo
theories~\citep{FredriksonJ14}.  None of these techniques can handle
advanced composition, interactive online algorithms and
privacy depending on accuracy.  \citet{BartheDGKB13} present a system
for reasoning about \emph{computational differential
  privacy}~\citep{MironovPRV09} a relaxation of differential privacy
where the adversary are computationally-bound.

Coupling is an established tool in probability theory, but it seems less
familiar to computer science. It was only quite recently that couplings have
been used in cryptography; according to \citet{HoangR10}, who use couplings to
reason about generalized Feistel networks, \citet{Mironov02} first used this
technique in his analysis of RC4. There are seemingly few applications of coupling in formal
verification, despite considerable research on probabilistic bisimulation (first
introduced by~\citet{LarsenS89}) and probabilistic relational program logics
(first introduced by~\citet{BartheGZ09}). The connection between
liftings and couplings was recently noted by~\citet{BartheEGHSS15} and
explored for differential privacy by~\citet{BGGHS16}. The latter
uses a coupling argument to prove differentially private the sparse
vector algorithm that we also consider in this work. The additional
challenges that we face are: first, the integration of advanced composition,
providing a much better privacy bound; second, the proof that sparse
vector is differentially private also in the interactive model, which
requires additionally to have a logic that permits to reason about the
adversary. Moreover, \citet{BGGHS16} do not provide methods to
prove privacy using accuracy.

In promising recent work, \citet{ZK16} design a system to automatically verify
differentially privacy for examples where the privacy proof uses tools beyond
the standard composition theorem, including the Sparse Vector technique.  Their
proof strategy is morally similar to couplings, but their work uses a
combination of product programs and lightweight refinement types backed by novel
type-inference techniques, rather than a relational program logic like we
consider. Their system can also optimize the privacy cost, something that we do
not consider. While their work is highly automated, their system is limited to
pure, $(\epsilon, 0)$ differential privacy, so it cannot verify the algorithms
we consider, where privacy follows from accuracy or the advanced composition
theorem. Their techniques also seem limited to modeling couplings that arise
from bijections; in particular, it is not clear how to prove privacy for
examples that use more advanced couplings like the optimal subset coupling.

\section{Concluding remarks} \label{sec:conclusions}
We have presented an extension of the logic \Saprhl~\citep{BartheKOZ13} that can
express three classes of privacy proofs beyond current state-of-the-art
techniques for privacy verification: privacy depending on accuracy, privacy from
advanced composition, and privacy for interactive algorithm. We have formalized
a generalization of the adaptive Sparse Vector algorithm, known as Between
Thresholds~\citep{BunSU16}. This and other possible generalizations of sparse
vector could bring interesting results in domains like
geo-indistinguishability~\citep{AndresBCP13}.

For the future, it would be interesting to explore generalizations of
differential privacy like the recent notion of \emph{concentrated differential
  privacy}~\citep{DworkR2016,BunS2016}.  This generalization features a simple
composition principle that internalizes the advanced composition principle of
standard differential privacy. However, it is currently unclear whether the
definition of concentrated differential privacy, which involves R\'enyi
divergences, can be modeled using \Saprhl.

Additionally, there is still room for improving the expressivity of
\Saprhl for differential privacy. One interesting
example combining accuracy and privacy is the \emph{large margin
mechanism}~\citep{ChaudhuriHS14}. The privacy proof for this algorithm
requires careful reasoning about the size of the support when applying pointwise
equality, and sophisticated facts about the accuracy Sparse Vector. This example
seems beyond the reach of our techniques, but we believe it could be handled by
generalizing the existing rules.

Finally, it would be interesting to explore a tighter integration of accuracy
and privacy proofs. We currently use two systems, \Sahl and \Saprhl, to verify
privacy. This can lead to awkward proofs since the two logics can only interact
in specific places in the proof (i.e., the up-to-bad rules). A combined
version of the logics could allow more natural proofs.

\bibliographystyle{abbrvnat}
\bibliography{header,main}

\iffull
\appendix

\section{Omitted proofs}

\propliftER*
\begin{proof}
  For the forward direction, take the witnesses $\mu_L, \mu_R \in \Dist(A \times
  A)$ and define witnesses $\nu_L \triangleq \lift{(f \times f)}(\mu_L)$ and
  $\nu_R \triangleq \lift{(f \times f)}(\mu_R)$. The support condition is clear,
  the marginal requirement is clear, and the distance requirement also follows
  directly: for any set $S \subseteq B \times B$, apply the distance condition
  on $\mu_L, \mu_R$ for the set $f^{-1}(S)$.

  For the reverse direction, let $\nu_L, \nu_R \in \Dist(B \times B)$ be the
  witnesses to the second lifting. We aim to construct a pair of witnesses
  $\mu_L, \mu_R \in \Dist(A \times A)$ to the first lifting. The basic idea is
  to define $\mu_L$ and $\mu_R$ based on equivalence classes of elements in $A$
  mapping to a particular $b \in B$, and then smooth out the probabilities
  within each equivalence class.
  To begin, for $a \in A$, define
  \begin{align*}
    [a]_{f} &\triangleq f^{-1}(f(a)) &
    \alpha_i(a) &\triangleq \textstyle \Pr_{\mu_i} [\{ a \} \mid [a]_f]
  \end{align*}
  We also assume a function $\rho$ that takes a set $S$ and selects an arbitrary
  representative $\rho(S) \in S$.  We take $\alpha_i(a) = 0$ when $\mu_i([a]_f)
  = 0$.

  Now, we define $\mu_L$ and $\mu_R$ as
  \begin{align*}
    \mu_L &: (a_1, a_2) \mapsto \alpha_L(a_1, a_2) \cdot \nu_L(f(a_1), f(a_2)) \\
    \mu_R &: (a_1, a_2) \mapsto \alpha_R(a_1, a_2) \cdot \nu_R(f(a_1), f(a_2))
  \end{align*}
  where
  \[
    \alpha_L(a_1, a_2) =
    \begin{cases}
      \alpha_1(a_1) \cdot \alpha_2(a_2) &: \mu_2([a_2]_f) \neq 0 \\
      \alpha_1(a_1) &: \mu_2([a_2]_f) = 0 \land a_2 = \rho([a_2]_f) \\
      0 &: \text{otherwise}
    \end{cases}
  \]
  \[
    \alpha_R(a_1, a_2) =
    \begin{cases}
      \alpha_1(a_1) \cdot \alpha_2(a_2) &: \mu_1([a_1]_f) \neq 0  \\
      \alpha_2(a_2) &: \mu_1([a_1]_f) = 0 \land a_1 = \rho([a_1]_f) \\
      0 &: \text{otherwise}
    \end{cases}
  \]
  %
  %


  The support and marginal conditions follow from the support and marginal
  conditions of $\nu_L$, $\nu_R$. For instance, note that if $\mu_1([a_1]_f) = 0$,
  then $(\pi_1 \mu_L)(a_1) = 0 = \mu_1(a_1)$. Otherwise, assume $\mu_1([a_1]_f)
  \neq 0$ and compute:
  \begin{align*}
    \sum_{a_2} \mu_L(a_1, a_2)
    &= \sum_{a_2} \alpha_L(a_1, a_2) \nu_L(f(a_1), f(a_2)) \\
    &= \sum_{a_2 : \mu_2([a_2]_f) = 0} \alpha_L(a_1, a_2) \nu_L(f(a_1), f(a_2))
    \\
    &+ \sum_{a_2 : \mu_2([a_2]_f) \neq 0} \alpha_L(a_1, a_2) \nu_L(f(a_1), f(a_2)) \\
    &= \sum_{[a_2]_f : \mu_2([a_2]_f) = 0} \alpha_1(a_1) \cdot \nu_L(f(a_1),
    f(a_2)) \\
    &+ \sum_{a_2 : \mu_2([a_2]_f) \neq 0}
    \alpha_1(a_1) \cdot \alpha_2(a_2) \cdot \nu_L(f(a_1), f(a_2)) \\
    &= \sum_{[a_2]_f : \mu_2([a_2]_f) = 0} \alpha_1(a_1) \cdot \nu_L(f(a_1),
    f(a_2)) \\
    &+ \sum_{[a_2]_f : \mu_2([a_2]_f) \neq 0}
    \alpha_1(a_1) \cdot \nu_L(f(a_1), f(a_2)) \\
    &= \alpha_1(a_1) \sum_{b_2 \in f(A)} \nu_L(f(a_1), b_2) \\
    &= \alpha_1(a_1) \mu_1([a_1]_f) = \mu_1(a_1) .
  \end{align*}
  The last step uses $f(A) = B$ because $f$ is surjective, so we may apply the
  marginal condition of $\pi_1 \nu_L = \lift{f}(\mu_1)$.  A similar calculation
  shows $\pi_2 \mu_R = \mu_2$.

  We now check the distance condition
  $\Delta_\epsilon(\mu_L, \mu_R) \leq \delta$.
  Since we have $\Delta_\epsilon(\nu_L, \nu_R) \leq \delta$, we know that
  \begin{gather*}
    \nu_L(b_1, b_2) \leq e^\epsilon \nu_R(b_1, b_2) + \delta(b_1, b_2)
  \end{gather*}
  for $\delta(b_1, b_2) \geq 0$ and $\sum_{b_1, b_2} \delta(b_1, b_2) \leq
  \delta$. We now prove that, for any $(a_1, a_2) \in A \times A$,
  we have
  $\mu_L(a_1, a_2) \leq e^\epsilon \mu_R(a_1, a_2) + \zeta(a_1, a_2)$
  where
  $\zeta(a_1, a_2) \triangleq
     \alpha_L(a_1,a_2) \cdot \delta(f(a_1), f(a_2))$.
  Let $(a_1, a_2) \in A \times A$ and $b_i \triangleq f(a_i)$.
  %
  %
  Then, consider $a_1 \in A$. If $\mu_1([a_1]_f) = 0$, we can immediately bound
  \[
    \mu_L(a_1, a_2) = 0 \leq e^\epsilon \mu_R(a_1, a_2) + \zeta(a_1, a_2) .
  \]
  Otherwise, assume $\mu_1([a_1]_f) \neq 0$. Then, we can bound:
  \begin{align*}
    \mu_L(a_1, a_2)
      &= \alpha_L(a_1, a_2) \cdot \nu_L(b_1, b_2) \\
      &\leq \alpha_L(a_1, a_2) \cdot (e^\epsilon \nu_R(b_1, b_2) + \delta(b_1, b_2)) \\
      &= e^\epsilon (\alpha_R(a_1, a_2) \cdot \nu_R(b_1, b_2)) +
           \alpha_L(a_1, a_2) \cdot \delta(b_1, b_2) \\
      &= e^\epsilon \mu_R(a_1, a_2) + \alpha_L(a_1, a_2) \cdot \delta(b_1, b_2) \\
      &\leq e^\epsilon \mu_R(a_1, a_2) + \zeta(a_1, a_2) .
  \end{align*}
  The third line changes from $\alpha_L$ to $\alpha_R$ in the first factor.
  Note that $\alpha_L(a_1, a_2) \neq \alpha_R(a_1, a_2)$ exactly when
  $\mu_2([a_2]_f) = 0$, when $\nu_R(b_1, b_2) = 0$ as well.

  Now, we just need to bound sums of $\zeta(a_1, a_2)$.  First, we can show that
  the sum of $\alpha_L$ within any equivalence class is at most $1$. Consider
  $b_1, b_2$. If $\mu_1(f^{-1}(b_2)) = 0$, then we have
  \begin{align*}
    \sum_{a_1 \in f^{-1}(b_1)} \sum_{a_2 \in f^{-1}(b_2)} \alpha_L(a_1, a_2)
    &= \sum_{a_1 \in f^{-1}(b_1)} \alpha_1(a_1) \leq 1 .
  \end{align*}
  Otherwise, if $\mu_1(f^{-1}(b_2)) \neq 0$, we have $\alpha_L(a_1, a_2) \leq
  \alpha_1(a_1) \alpha_2(a_2)$ and again
  \begin{align*}
    \sum_{a_1 \in f^{-1}(b_1)} \sum_{a_2 \in f^{-1}(b_2)} \alpha_L(a_1, a_2) \leq 1 .
  \end{align*}
  Hence,
  \begin{align*}
    \sum_{a_1, a_2} \zeta(a_1, a_2) &=
      \sum_{b_1, b_2} \delta(b_1, b_2)
        \sum_{a_1 \in f^{-1}(b_1)}
        \sum_{a_2 \in f^{-1}(b_2)} \alpha_L(a_1, a_2) \\
      &\leq \sum_{b_1, b_2} \delta(b_1, b_2) \leq \delta.
  \end{align*}
  So for any $S \subseteq A \times A$, we have $\mu_L(S) \leq e^\epsilon
  \mu_R(S) + \delta$. This shows that  $\Delta_\epsilon(\mu_L, \mu_R) \leq 0$,
  as desired.
\end{proof}

\lemlapint*
\begin{proof}
  We first define some notation. Let $W$ be the total mass of the Laplace
  distribution, without normalization. By direct calculation,
  \[
    W = \sum_{r = -\infty}^{+\infty} \exp(- |r| \epsilon) = \frac{e^\epsilon +
      1}{e^\epsilon - 1} .
  \]
  Let $L(x, y)$ be the mass of the Laplace distribution in $[x, y]$ for $x, y
  \leq 0$. Again by direct calculation,
  \[
    L(x, y) = \frac{1}{W} \sum_{r = a}^b \exp(r \epsilon) = \frac{ e^{(y + 1)
        \epsilon}
      - e^{x \epsilon}}{e^\epsilon + 1} .
  \]

  With this notation, we want to show $L(a', b') \leq \alpha L(a, b)$.  There
  are four cases: $a < b \leq 0$, $a < 0 < b$ with $|a| \leq |b|$, $0 \leq a <
  b$, and $a < 0 < b$ with $|a| \geq |b|$. By symmetry of the Laplace
  distribution it suffices to consider the first two cases.

  For the first case, $a < b \leq 0$. By direct calculation:
  \begin{align*}
    L(a', b') &\leq L(a, b) + \frac{1}{W} \sum_{r = b + 1}^{b + \eta} e^{r
    \epsilon} \\
    &= \frac{ e^{ (b + 1 + \eta)\epsilon }  - e^{ a \epsilon } }{ e^{\epsilon} + 1 }
    \\
    &= \frac{1}{e^{\epsilon} + 1} (e^{ (b + 1)\epsilon }  - e^{ a \epsilon })
    \left( \frac{ e^{\eta \epsilon} - e^{- (b - a + 1) \epsilon } }{ 1 - e^{ - (b
          - a + 1) \epsilon }} \right)
    \\
    &= \left( \frac{ e^{ \eta \epsilon } - e^{ - (b - a + 1) \epsilon } }{ 1 -
        e^{ - (b - a + 1) \epsilon }} \right) L(a, b)
    \leq \alpha L(a, b)
  \end{align*}

  For the second case, $a < 0 < b$ with $|a| \leq |b|$. By direct calculation:
  \begin{align*}
    L(a', b') &\leq L(a, b) + \eta L(a, a)
    = L(a, b) + \eta \left( \frac{e^{ \epsilon } - 1}{e^{ \epsilon } + 1}
    \right) e^{ a \epsilon }
    \\
    &\leq \left( 1 + \eta \left( \frac{e^{ \epsilon } - 1}{e^{ \epsilon } + 1}
      \right)  \frac{e^{ a\epsilon }}{L(0, b)} \right) L(a, b)
    \\
    &= \left( 1 + \frac{\eta (e^{ \epsilon } - 1) e^{ a \epsilon }}{ e^{ \epsilon } -
        e^{ - b \epsilon } } \right) L(a, b)
    \\
    &= \left( \frac{1 - e^{ -(b + 1)\epsilon } + \eta(e^{ \epsilon } - 1)
        e^{ (a - 1) \epsilon } }{ 1 - e^{ -(b + 1)\epsilon } } \right) L(a, b)
    \\
    &\leq \left( \frac{ 1 + \eta (e^{ \epsilon } - 1)}{1 - e^{ - (b + 1) \epsilon }}
    \right) L(a, b)
    \\
    &\leq \left( \frac{ e^{ 2 \eta \epsilon } }{1 - e^{ - (b - a + 2)\epsilon/2 }} \right)
    \leq \alpha
  \end{align*}
  The last line is because $(b + 1) \geq (b - a + 2)/2$, and because $1 + \eta
  (e^{ \epsilon } - 1) \leq e^{ \eta \epsilon }$ for $\eta \in \mathbb{Z}$ for all
  $\eta, \epsilon \geq 0$.  To see this last fact, note that the inequality holds
  (with equality) at $\eta = 0$. We can directly check that
  \[
    \frac{1 + (\eta + 1) (e^{ \epsilon } - 1)}{1 + \eta (e^{ \epsilon } - 1)}
    \leq
    \frac{e^{ (\eta + 1) \epsilon }}{e^{ \eta \epsilon }} = e^{ \epsilon } ,
  \]
  for $\epsilon \geq 0$, so the inequality is preserved as we increase $\eta$.
\end{proof}
\fi

\end{document}